%% file: main_revised1.tex
\documentclass[preprint,12pt]{elsarticle}
\usepackage{amsmath,amsfonts,amssymb}
\usepackage{algorithm}
\usepackage{algpseudocode}
\usepackage{array}
\usepackage{booktabs}
\usepackage{longtable}
\usepackage{multirow}
\usepackage{tabularx}
\usepackage{graphicx}
\usepackage{subcaption}
\usepackage{textcomp}
\usepackage{xcolor}
\usepackage{dirtytalk}
\usepackage{csquotes}
\usepackage{pdflscape}
\usepackage{multicol}
\usepackage{stfloats}
\usepackage[many]{tcolorbox}
\usepackage{amsthm}
\usepackage{tikz}
\usetikzlibrary{turtle}
\usepackage{url}
\usepackage{acro}
\usepackage{verbatim}
\usepackage{cuted}
\usepackage{float}
\usepackage[export]{adjustbox}
\usepackage{mathtools}
\usepackage{xurl}

\newcommand{\inlineeqnum}{\refstepcounter{equation}~~\mbox{(\theequation)}}

\newtcbtheorem[auto counter]{examplebox}{Example}{
    boxrule = 1.5pt,
    colframe = black,
    fonttitle = \bfseries
}{ex}

\newtheorem{lemma}{Lemma}
\newtheorem{proposition}{Proposition}
\newtheorem{theorem}{Theorem}
\newtheorem{sublemma}{Lemma}[lemma]

\newcounter{rqcounter}

\newcommand*\circled[1]{\tikz[baseline=(char.base)]
  {\node[shape=circle,fill,inner sep=1.5pt]
   (char) {\textcolor{white}{#1}};}}

\NewDocumentCommand{\rot}{O{45} O{1em} m}{%
  \makebox[#2][l]{\rotatebox{#1}{#3}}%
}

\DeclareAcronym{pbs}{
  short = PBS,
  long = Proposer Builder Separation
}
\DeclareAcronym{epbs}{
  short = ePBS,
  long = enshrined Proposer Builder Separation
}
\DeclareAcronym{mev}{
  short = MEV,
  long = Maximal Extractable Value
}
\DeclareAcronym{defi}{
  short = DeFi,
  long = decentralized finance
}
\DeclareAcronym{tee}{
  short = TEE,
  long = Trusted Execution Environments
}
\DeclareAcronym{dex}{
  short = DEX,
  long = decentralized exchanges
}
\DeclareAcronym{abm}{
  short = ABM,
  long = agent-based modelling
}
\DeclareAcronym{pow}{
  short = PoW,
  long = Proof-of-Work
}
\DeclareAcronym{pos}{
  short = PoS,
  long = Proof-of-Stake
}
\DeclareAcronym{bne}{
  short = BNE,
  long = Bayesian Nash Equilibrium
}

\usepackage[unicode]{hyperref}


\begin{document}

\begin{frontmatter}

\title{Enshrined Proposer Builder Separation in the presence of Maximal Extractable Value}

\input{titlepage_revised1}

\begin{abstract}

In blockchain systems operating under the \acf{pos} consensus mechanism, fairness in transaction processing is essential to preserving decentralization and maintaining user trust. However, with the emergence of \acf{mev}, concerns about economic centralization and content manipulation have intensified. To address these vulnerabilities, the Ethereum community has introduced \acf{pbs}, which separates block construction from block proposal. Later, \acf{epbs} was also proposed in EIP-7732, which embeds \acs{pbs} directly into the Ethereum consensus layer. 

Our work identifies key limitations of \acs{epbs} by developing a formal framework that combines mathematical analysis and agent-based simulations to evaluate its auction-based block-building mechanism, with particular emphasis on \acs{mev} dynamics. Our results reveal that, although \acs{epbs} redistributes responsibilities between builders and proposers, it significantly amplifies profit and content centralization: the Gini coefficient for profits rises from 0.1749 under standard \acs{pos} without \acs{epbs} to 0.8358 under \acs{epbs}. This sharp increase indicates that a small number of efficient builders capture most value via \acs{mev}-driven auctions. Moreover, 95.4\% of the block value is rewarded to proposers in \acs{epbs}, revealing a strong economic bias despite their limited role in block assembly. These findings highlight that \acs{epbs} exacerbates incentives for builders to adopt aggressive \acs{mev} strategies, suggesting the need for future research into mechanism designs that better balance decentralization, fairness, and \acs{mev} mitigation.

\end{abstract}


\end{frontmatter}

\section{Introduction}
\label{sec:intro}

Ethereum's transition to Proof-of-Stake (PoS) represents a major shift in its consensus mechanism, but has not eliminated economic vulnerabilities. A critical concern arises from \acf{mev}, the value extractable through the inclusion, exclusion, or reordering of transactions. \ac{mev} creates strong incentives for manipulation, and evidence shows that between late 2023 and early 2024, just three builders produced nearly 80\% of all Ethereum blocks \cite{Daian2019, Oz2024WhoWhy}. Such concentration threatens decentralization, fairness, and credibility of the consensus protocol.

To address these issues, \acf{pbs} was proposed by Buterin \cite{ethpbs}. By isolating the roles of proposers and builders, \ac{pbs} seeks to reduce validators' discretion over transaction ordering and promote a more balanced market. Its protocol-native variant, \acf{epbs}, has since been drafted in EIP-7732 \cite{eip7732}, aiming to eliminate dependence on third-party relays and mitigate censorship concerns \cite{Neuder2023_epbs}. The Ethereum community is currently debating \ac{epbs} as a fundamental redesign of the consensus layer \cite{eip7732discussion}.

Yet the core question remains unresolved: \emph{does \ac{epbs} actually mitigate centralization and manipulation, and/or does it introduce new risks}? Existing studies \cite{Heimbach2023, Wahrstatter2023, Bahrani2024} primarily focus on empirical observations of relay-based \ac{pbs} and do not provide a formal framework to analyze the system-level incentives under \ac{epbs}. As a result, the long-term implications of embedding \ac{pbs} into Ethereum's protocol remain unclear.

In this work, we conduct a comprehensive study of \ac{epbs}, combining formal modeling with mathematical analysis and agent-based simulations. Our framework represents the monetary utility of users, builders, and proposers, enabling us to evaluate how \ac{mev}-driven auctions' equilibrium outcomes. We focus on two research questions:  

\textbf{RQ1: What is the effect of \ac{epbs} on fairness and decentralization?} We examine (1) how block construction decentralization differs between \ac{pos} and \ac{epbs}, (2) how the profit distribution differs between builders and validators under the two systems, and (3) how profit restaking and accumulation affect long-term concentration.  

\textbf{RQ2: What is the effect of \ac{epbs} on transaction manipulation?} We investigate whether the auction-based design reduces transaction reordering, or instead amplifies \ac{mev}-driven incentives that harm user fairness.

\paragraph*{Our Contributions}  
This paper makes the following contributions:
\begin{enumerate}
    \item We develop a formal framework of users, builders, and proposers in \ac{epbs}, and use it to analyze the single block auction dynamics and long-term effects. Our theoretical results show that \ac{mev}-seeking builders consistently dominate the auctions, leading to highly unequal profit distributions and long-term centralization effects through re-staking.  
    \item We implement calibrated agent-based simulations using Ethereum on-chain data, demonstrating that while proposer revenue becomes more decentralized under \ac{epbs}, builders capture an overwhelming share of \ac{mev}, amplifying market concentration and transaction reordering compared to \ac{pos}.  
    \item We identify a fundamental limitation of \ac{epbs}: although it achieves decentralization at the proposer level, it intensifies centralization among builders and worsens \ac{mev}-driven manipulation, raising concerns for Ethereum's long-term fairness and security.  
\end{enumerate}

\section{Background}
\label{sec:background}

In this section, we elaborate on the background knowledge related to \ac{epbs}, including its consensus setting, protocol details, and prior studies.

\begin{figure*}[tb]
    \centering
    \begin{subfigure}[t]{0.48\textwidth}
        \centering
        \includegraphics[height=7cm, keepaspectratio, trim=30 20 30 10]{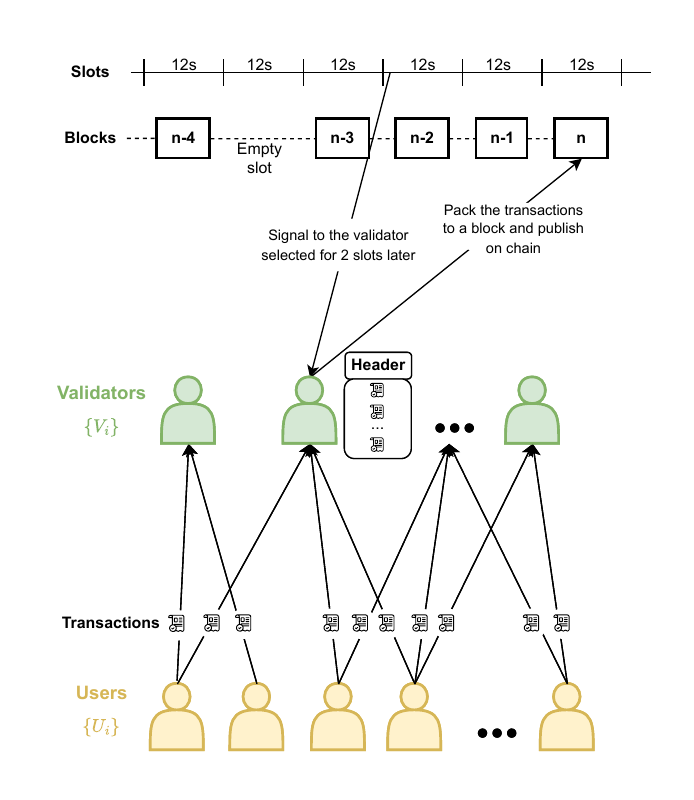}
        \caption{PoS.}
        \label{fig:pos_workflow}
    \end{subfigure}
    \hfill
    \begin{subfigure}[t]{0.48\textwidth}
        \centering
        \includegraphics[height=7cm, keepaspectratio, trim=30 20 30 10]{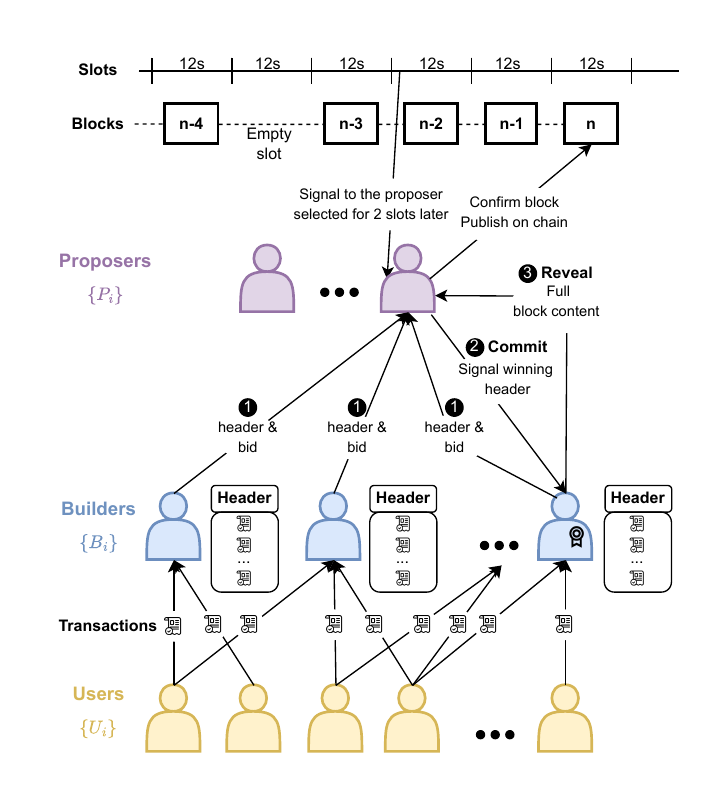}
        \caption{ePBS.}
        \label{fig:pbs_workflow}
    \end{subfigure}
    \caption{Agent workflow where each arrow represents message propagation, described with more details in \autoref{sec:pbs_mech}.}

\end{figure*}

\subsection{Current Ethereum Consensus Mechanism: PoS}
\label{sec:eth_roadmap}
Since Ethereum's consensus layer transitioned to \ac{pos} \cite{merge}, only nodes that have deposited 32 ETH are eligible to become validators. A validator is responsible for building a block during a specific time window called a slot---a fixed, continuous 12-second interval. Validators are randomly chosen two slots in advance from eligible participants to propose a block. If the assigned validator fails to propose a block, the slot remains empty, leading to a small percentage of skipped slots \cite{Beaconscan2025SkippedBlocks}. This process is illustrated in \autoref{fig:pos_workflow}. 

Ethereum's consensus mechanism faces challenges and concerns such as \ac{mev}, transaction censorship, and centralization \cite{Motepalli2025DecentralizationAdvancement}, which we further discuss below.

\subsection{PBS and ePBS}
\label{sec:pbs_mech}
\ac{pbs} is an Ethereum-based mechanism aiming to increase the decentralization of profit and block building \cite{pbsdef} by decoupling block construction from block proposal. In its current deployment of relay-based \ac{pbs} or MEV-Boost \cite{flashbot}, these roles interact through trusted off-chain relays that collect bids from builders and forward only the winning header to the proposer. This architecture employs an off-chain commit–reveal scheme, where builders first send a block header and bid to the relay, and later reveal the full block content once the proposer commits. While effective in preventing proposers from manipulating transactions, this design introduces new trust assumptions: relays may censor bids, fail to deliver payloads, or extract informational rents, and such behavior is unverifiable at the protocol level \cite{Heimbach2023, Wahrstatter2023}. 

\ac{epbs}, specified in EIP-7732~\cite{eip7732}, integrates the proposer–builder interaction directly into Ethereum's consensus layer, removing the need for off-chain relays. Bids, commitments, and payments are exchanged through protocol-defined messages that are verifiable on chain, guaranteeing proposer payment and preventing builders from withholding payloads or altering bids after commitment. By decoupling header attestation from full payload execution, \ac{epbs} improves liveness under network latency and enforces builder payments at the consensus level. In essence, \ac{epbs} replaces the economically trusted relay model of MEV-Boost with a protocol-native, cryptographically accountable mechanism while maintaining the same first-price auction structure~\cite{Neuder2023BidHarmful}.

A key feature of the \ac{pbs} and \ac{epbs} design is the commit–reveal scheme. As illustrated in \autoref{fig:pbs_workflow}, in step \circled{1}, builders transmit block headers and bids without revealing transaction order. In step \circled{2}, the proposer commits to a header. Finally, in step \circled{3}, the builder reveals the full block content. Once committed, neither the proposer nor the builder can alter or reorder transactions. Builders submit bids every 0.5 seconds, i.e., 24 rounds per slot \cite{2subsec}. Proposers are randomly selected per slot, consistent with validator selection in PoS. 

\subsection{Maximal Extractable Value (MEV)}
\label{sec:mev}
Ethereum's PoS design allows validators to determine which transactions to include and in what order, enabling the extraction of \ac{mev} through inclusion, exclusion, or reordering~\cite{ethmev}. When the \ac{mev} of a block exceeds the base block reward, validators may exclude transactions with MEV potential and replace them with their own strategically ordered ones~\cite{Daian2019}. 

Typical \ac{mev} strategies include front-running~\cite{Eskandari2019}, back-running~\cite{Werner2022SoK:DeFi}, sandwich attacks~\cite{Zust2021}, censorship and omission~\cite{Wahrstatter_Censorship}, and arbitrage~\cite{zhou_arbitrage}. While some forms, such as arbitrage, improve efficiency by aligning DEX prices~\cite{Makarov2019}, others (e.g., sandwiching) degrade user experience through slippage and increased confirmation time.

\subsection{Existing Literature}
\label{sec:lit}
Prior work has examined challenges in Ethereum consensus, PBS, and MEV. For PBS, Heimbach et al.~\cite{Heimbach2023} studied decentralization across relays; Gupta et al.~\cite{Gupta2023} modeled private order flow auctions, showing integrated builder–searchers gain persistent advantages; Bahrani et al.~\cite{Bahrani2024} provided a theoretical model of MEV-driven centralization; Yang et al.~\cite{Yang2024} conducted an empirical study of the builder market; Wang et al.~\cite{Wang2024} developed an asymmetric auction model; and Capponi et al.~\cite{Capponi2024} used game theory to show reduced validator centralization but heightened builder concentration. Improvement proposals include Buterin's partial block auctions~\cite{Buterin2022} and Babu's modified PBS model~\cite{Babu2022}, though both lacked validation via simulations.

Related work on MEV and transaction-level dynamics includes Grandjean et al.~\cite{Grandjean2023}, who identified decentralization deficiencies in Ethereum PoS, and Wahrstätter et al.~\cite{Wahrstatter2023}, who analyzed the Ethereum block construction market, revealing intricate \ac{mev} dynamics, with their follow-up~\cite{Wahrstatter_Censorship} showing censored transactions suffer an average delay of 85\%. Early studies by Daian et al.~\cite{Daian2019} and Eskandari et al.~\cite{Eskandari2019} formalized \ac{mev}, highlighting bot-driven strategies that distort fairness. Subsequent work quantified MEV, e.g., Qin et al.~\cite{Qin2021} and Weintraub et al.~\cite{Weintraub2022}, while Chitra et al.~\cite{Chitra2022} proposed redistribution mechanisms. Yang et al.~\cite{yang2022} catalogued countermeasures, and Jensen et al.~\cite{Jensen2023} studied multi-block MEV dynamics.

Overall, while both empirical and theoretical studies provide insight, they largely focus on relay-based PBS or isolated auction models. Little work addresses \ac{epbs} directly or provides system-level formalization combined with simulations. Our study fills this gap by modeling the strategic incentives of users, builders, and proposers under \ac{epbs}, and analyzing decentralization and MEV dynamics through both mathematical analysis and simulation.

\section{Formalization}
\label{sec:formalization}

In our \ac{epbs} model, we consider three classes of agents acting within a single block: users \{\(U_i\)\}, builders \{\(B_i\)\}, and proposers \{\(P_i\)\} also labeled in \autoref{fig:pbs_workflow}. The actions are defined in \autoref{sec:actions} and the utilities of each class of agent are formalized in \autoref{sec:utility}. The decision-making process for a block is illustrated in \autoref{fig:bid_rounds}. 
Our analysis focuses exclusively on monetary gains and losses, particularly those related to \ac{mev} extraction and block production rewards. All agent utilities are defined in terms of direct economic incentives measured. 
For the formalization of single-block dynamics, we drop the slot subscript $\ell$ for clarity, as all analysis focuses on actions within a single block.

\subsection{Assumptions}
\label{sec:assumptions}
\small
We list the following modeling assumptions:

\begin{enumerate}
    \item Participants are fully described by the attribute tuples provided in \autoref{sec:system-protocol}, without additional side information affecting their decisions.
    
    \item Latency in the network is modeled by an undirected weighted graph \(\mathcal{G} = (\mathcal{Z}, \mathcal{E})\). Each edge \( e_{uv} \in \mathcal{E}\) has weight \(w(e_{uv})\), indicates the link latency in the unit of rounds, and \(d(i,j) = d(j,i)\) represent the shortest distance between node \(z_i\) and node \(z_j\)~\cite{Jayabalasamy2024}. 

    Due to non-instantaneous propagation, both mempool contents and bid histories are affected by latency. 

    \begin{examplebox}{}{network}
    \small
    In the example \autoref{fig:network}, multiple paths between \( z_i \) and \( z_j \) exist: direct link  \( z_i \) and \( z_j \), via nodes \( z_k \), \( z_h \), or \( z_l \). The shortest latency \( d(i,j) \) corresponds to the minimum weight over all such paths: 
    \begin{align*}
        d(i,j) = \min\Big\{ &\textcolor{green!60!black}{w(e_{ik}) + w(e_{kj})},\, \textcolor{blue}{w(e_{ij})},
        \textcolor{red}{w(e_{ih}) + w(e_{hl}) + w(e_{lj})} \Big\}.
    \end{align*}

    \begin{figure}[H]
        \centering
        \includegraphics[width=0.3\linewidth,trim=20 15 20 30]{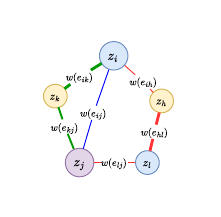}
        \caption{Example of a graph network with weight of edge representing latency.}
        \label{fig:network}
    \end{figure}
    \end{examplebox}

    \item Auctions proceed in discrete rounds with a fixed maximum round at \(T \leq 24\). This is an assumption for the simplification of the model, as builders are only allowed to submit a bid every 0.5 second, so 24 bids would be the maximum that any builder can submit.
    
    \item The cost of participating is negligible. If a transaction from the user is not included in the block or the block packed by the builder did not win the auction, no utility is deducted.

\end{enumerate}

\subsection{System Protocol}
\label{sec:system-protocol}

The protocol components are described as follows.

\begin{enumerate}

    \item Each validator \( V_i \) at round \(t\) has an attribute tuple \( (\tau_{V_i}, \mathcal{M}_{V_i, t}, \mathbf{x}_{V_i, t}) \), where \( \tau_{V_i} \in \{\mathtt{benign}, \mathtt{attack}\} \) indicates whether they behave honestly or attempt to exploit \ac{mev}, \( \mathcal{M}_{V_i, t} \) is the mempool of \( V_i \) at round \( t \), and \( \mathbf{x}_{V_i, t} \) is a list of transactions submitted by \( V_i \) up to round \( t \). The attribute \(\tau\) is explained in more detail in \autoref{sec:actions}.
    
    \item Each builder \( B_i \) at round \(t\) has an attribute tuple \( (\tau_{B_i}, \sigma_{B_i},  \mathcal{M}_{B_i, t}, \mathbf{x}_{B_i, t} )\), extending the validator attributes with \( \sigma_{B_i} \) (later referred to as \(\sigma_i\)) as their bidding strategy, which determines their bid in each auction round. 
    
    \item Each user \( U_i \) at round \(t\) has an attribute tuple \( (\tau_{U_i}, \mathcal{M}_{U_i, t}, \mathbf{x}_{U_i, t}) \), where \( \tau_{U_i} \in \{\mathtt{benign}, \mathtt{attack}\} \) indicates whether they behave honestly or attempt to exploit \ac{mev}, \(\mathbf{x}_{U_i, t}\) is the list of transactions submitted by the user \( U_i \) up to round \(t\), and \( \mathcal{M}_{U_i, t} \) is the mempool of \(U_i\) at round \(t\). Benign users submit standard transactions with no specific target, whereas attacking users strategically adjust transaction fees to manipulate execution order and maximize \ac{mev} extraction from target transactions.

    \item Each builder \(B_i\) has their private valuation \(v_{i,t}\) for the list of ordered transactions \(\mathcal{T}_{i,t}\) selected for round \(t\).
    
    \item The mempool \(\mathcal{M}_{i,t}\) observed by node \(i\) at round \(t\) includes only transactions that could have arrived within \(t\) rounds under the delay constraints in \(\mathcal{G}\), as illustrated in \autoref{fig:network}.
    
    \item Similarly, the bid history \(H_{i,t}\) available to builder \(B_i\) at round \(t\) consists of all bids that could have reached \(i\) by round \(t\):
    \(H_{i,t} = \bigcup_{j} \left\{ b_{j,k} \mid d(i,j) + k \leq t \right\}.
    \)
    This reflects that builders operate with delayed and incomplete information, both in terms of transactions and bids.

    \item For any transaction \(x_j \in \mathbf{x}_{U_i,t}\) created by user \(U_i\) regardless of its characteristics attribute \(\tau_{U_i}\), there is a chance that it would have \ac{mev} opportunities \(m_j \geq 0\) for others to attack and gain profit from.

    \item The proposers in \ac{epbs} are always $\mathtt{benign}$ unless vertically integrated with builders as determined by the mechanism of not able to see or reorder the transactions before committing to one block.
\end{enumerate}

\begin{table}[tb]
\scriptsize

\caption{Mathematical notations.}
\begin{tabular}{l p{0.8\columnwidth}}
\toprule
\textbf{Notation}  & \textbf{Definition}       \\
\midrule
\(U\)      & User                       \\
\(B\)      & Builder                     \\
\(P\)      & Proposer                 \\
\(V\)      & Validator                \\
\(b\)      & Bid in gwei, \(b \geq 0\)   \\
\(x\)      & Transaction              \\
\(g\)      & Gas fee in gwei, \(g \geq 0\)       \\
\(m\)      & MEV profit potential in gwei, \(m \geq 0\) (\(m_j\) is the profit an attacking participant can get if they successfully launched an attack targeting at \(x_j\))     \\
\(\mathbf{x}\)     & Set of unordered transactions    \\
\(\mathcal{T}\)    &  List of ordered transactions    \\
\(v\)      & Private valuation in gwei       \\
\(s\)      & Stake          \\
\(\ell\)   & Slot number   \\
\(H\)      & History of all observed bids       \\
\(\tau\)   & Type of agent        \\
\(\sigma\) & Bidding strategy          \\
\(\mathcal{A}\)    & One-to-one mapping of the attack transaction to its target transaction  \\
\(\mathcal{M}\)    & Mempool content \\
\(\mathcal{U}\)    & Utility function in gwei \\
\bottomrule
\end{tabular}
\end{table}

\begin{figure}[H]
    \centering
    \includegraphics[width=0.8\linewidth,trim=30 20 30 20]{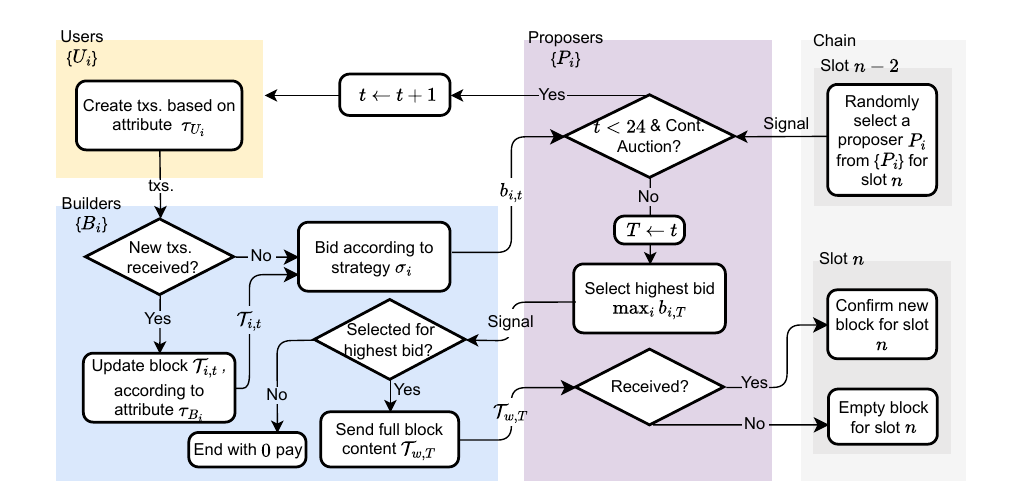}
    \caption{Decision cycle of ePBS.}
    \label{fig:bid_rounds}
\end{figure}

\subsection{Actions}
\label{sec:actions}

In \ac{epbs}, each agent type has a set of possible actions that they can take. Below, we outline the possible actions for users, builders, and proposers. For the formalization, we will focus on single block dynamics, thus all the slot subscripts $\ell$ are dropped.

\subsubsection{User Actions}
\paragraph*{Creating An Attack Transaction} If \(\tau_{U_i} = \mathtt{attack}\), user \(U_i\) initiates a transaction \(x_k\) targeting another transaction \(x_j = \mathcal{A}(x_k)\) where \(x_j \in  \mathcal{M}_{U_i, t}\) aimed at exploiting \ac{mev}, achieved by setting the gas fee \(g_k\) such that \(g_k \gtrless g_j\) for front and back running respectively. Note that this does not mean the attack is successful; the settlement happens when this transaction is confirmed on chain.

\paragraph*{Creating A Benign Transaction} If \(\tau_{U_i} = \mathtt{benign}\), user \(U_i\) initiates a transaction without targeting another transaction for \ac{mev} profits.

All user-initiated transactions can be attacked by other attacking users or builders if there is \ac{mev} potential available to be exploited.

\subsubsection{Builder Actions}

\paragraph*{Transaction Ordering}

Builder \(B_i\) selects an ordered list of transactions \(\mathcal{T}_{i,t} \subseteq \mathcal{M}_{B_i,t}\) to include in its block at round \(t\). If \( \tau_{B_i} = \mathtt{attack} \), the selection involves inserting a list of transactions \( \mathbf{x}_{B_i,t} \) that exploit \ac{mev} opportunities whenever profitable. Specifically, this follows the logic of inserting their own 0-gas attack transaction if the profit from the attack exceeds the total gas fees from any transactions displaced to accommodate the attack. Otherwise, the transaction is included as it is. If \( \tau_{B_i} = \mathtt{benign} \), the selection will be only to select transactions from the mempool and order them in descending order of gas fees. 

\begin{examplebox}{}{attack}
\small
Consider a block with capacity for three transactions. The mempool contains:
\begin{itemize}
    \item \(x_1\): transaction with gas fee \(g_1 = 30\) gwei
    \item \(x_2\): transaction with gas fee \(g_2 = 20\) gwei
    \item \(x_3\): transaction with gas fee \(g_3 = 15\) gwei
    \item \(x_4\): victim transaction with MEV opportunity \(m_4 = 50\) gwei, gas fee \(g_4 = 5\) gwei
\end{itemize}
A benign builder would select the top three highest-fee transactions:
\(
30 + 20 + 15 = 65 \text{ gwei}.
\)
In contrast, an attack builder considers inserting a 0-gas attack targeting \(x_4\), including \(x_4\), and one honest transaction (say \(x_1\)), displacing the others. The resulting block includes:
\begin{itemize}
    \item attack transaction: 0 gwei, gains MEV \(m_4 = 50\) gwei
    \item victim \(x_4\): 5 gwei
    \item \(x_1\): 30 gwei
\end{itemize}
with total block value
\(
50 + 5 + 30 = 85 \text{ gwei}.
\)
Since \(85 > 65\), the attack builder prefers to insert the attack. If instead the MEV were only \(m_4 = 10\) gwei, then
\(
10 + 5 + 30 = 45 \text{ gwei}
\)
would be worse than 65, so even the attack builder would not attempt it. This illustrates that the builder only inserts an attack if its MEV profit exceeds the total gas fees of the transactions displaced to accommodate it.

\end{examplebox}

\paragraph*{Bidding}
Builder \(B_i\) Submitting a bid \(b_{i,t }\) to the proposer according to the strategy \(\sigma_i (v_{i,t}, H_{i,t})\), where \(v_{i,t}\) is the valuation of the block at round \(t\).

\subsubsection{Proposer Actions}

\paragraph*{Auction Decision}
In the ideal case of \ac{epbs} where builders and proposers are separate nodes, the proposer proposes the highest bidding block. Proposer will stop the auction at any round \(T \leq 24\) and accept the block with the highest bid \( \max_{n} b_{n,T}\) at the stopping round \(T\).

\subsection{Utility Functions}
\label{sec:utility}

In this section, we define the utility functions for different agents participating in the \ac{epbs} system. These utility functions characterize the monetary economic incentives driving the actions of users, builders, and proposers. It is important to note that \ac{mev} is a zero-sum game, all gains for attackers result in losses for other participants.

\subsubsection{User Utility}
Let \( x_j \in \mathbf{x}_{U_i} \) be a transaction submitted by user \( U_i \), and let \( \mathcal{T}_{w,T} \) denote an ordered list of transactions in the final block selected at round \( T \) built by the winning builder \( B_w \). Let \( \mathcal{A}(x_k) \) denote the transaction targeted by the attacking transaction \( x_k \). Then, the utility of user \( U_i \) is given by

{\footnotesize
\begin{equation}
\label{eq:user_utility}
\mathcal{U}_U(U_i) =
\overbrace{\sum_{x_j \in (\mathbf{x}_{U_i} \cap \mathcal{T}_{w,T})}}^{\text{\scriptsize \shortstack{All user txs\\in final block}}}
\!\left[
\overbrace{\sum_{x_k \in \mathcal{T}_{w,T}}}^{\text{\scriptsize \shortstack{All txs\\in final block}}}
\!\left(
\overbrace{\mathbf{1}_{x_k = \mathcal{A}(x_j)} \cdot m_k}^{\text{\scriptsize \shortstack{Gain if $x_j$\\attacks $x_k$}}}
-
\overbrace{\mathbf{1}_{x_j = \mathcal{A}(x_k)} \cdot m_j}^{\text{\scriptsize \shortstack{Loss if $x_j$\\attacked}}}
\right)
-
\overbrace{g_j}^{\text{\scriptsize Gas fee}}
\right].
\end{equation}
}

Note that all utilities are defined on a per‐block basis.
This function captures the user's potential gain and loss due to \ac{mev}, reflecting the fact that \ac{mev} is a zero-sum game: any value extracted by one participant must have been taken from others within the system.

\subsubsection{Builder Utility}
\label{eq:builder_utility}

Let \( B_i \) be a builder participating in the auction at round \( t \), submitting bid \( b_{i,t} \), and let \( v_{i,t} \) be the corresponding block valuation. We distinguish between:

\paragraph*{Expected Utility (at round \( t \))}  
The builder's expected utility of the auction termination round $T$ at current round $t$ is: 
\begin{equation}
\mathbb{E}[\mathcal{U}_B(B_i, T)|t] = \sum_{z=t}^{24} [\mathbb{P}[T=z] \mathbb{P}[b_{i,z} = \max_n b_{n,z}] (v_{i,z} - b_{i,z})].
\label{eq:builder_expected_utility}
\end{equation}

\paragraph*{Realized Utility (at round \( T \))}  
Once the proposer ends the auction at round \( T \), the realized utility is  
\begin{equation}
\mathcal{U}_B(B_i, T) \;=\;
\begin{cases}
    v_{i,T} - b_{i,T}, & \text{if } b_{i,T} = \max_{n}\,b_{n,T},\\
    0, & \text{otherwise}.
\end{cases} 
\label{eq:builder_realized_utility}
\end{equation}

We assume that there are no 2 equal maximum bids due to the low probability.

\paragraph*{Block Valuation}
\label{sec:valuation}

The valuation of the block built by builder \( B_i \) at time \( t \) is given by
\begin{equation}
\label{eq:builder_valuation}
v_{i,t} =
\overbrace{\sum_{x_j \in \mathcal{T}_{i,t}}}^{\text{\scriptsize \shortstack{All txs\\in the block}}} 
\!\left[
\overbrace{\sum_{x_k \in (\mathcal{T}_{i,t} \cap \mathbf{x}_{B_i,t})}}^{\text{\scriptsize \shortstack{All builder\\initiated txs}}} 
\!\left(
\overbrace{\mathbf{1}_{x_j = \mathcal{A}(x_k)} \cdot m_j}^{\text{\scriptsize \shortstack{Gain if $x_j$\\attacked by $x_k$}}}
\right)
+
\overbrace{g_j}^{\text{\scriptsize Gas fee}}
\right],
\end{equation}
where \( \mathcal{T}_{i,t} \) is the set of transactions included by builder \( B_i \) in its block at round \( t \), and \( \mathbf{x}_{B_i,t} \) is the set of transactions created by builder \( B_i \) up to time \( t \). The second term aggregates the MEV value extracted by any attacking transaction submitted by \( B_i \) that is included in the final block.

\paragraph*{Builder's Bid}

The bid \(b_{i,t}\) submitted by builder \(B_i\) at round \(t\) is determined by the builder's strategy \(\sigma_i\),
$b_{i,t} \;=\; \sigma_i\bigl(v_{i,t}, H_{i,t}\bigr) \inlineeqnum \label{eq:builder_bid}$,
where \(v_{i,t}\) is the builder's block valuation at round \(t\), and \(H_{i,t}\) is the history of all bids that \(B_i\) observes at round \(t\).

\subsubsection{Proposer Utility}
Let \(P_i\) be the proposer that is selected to perform the block validation. If the proposer terminates the auction at round \(T\), then the proposer's utility is given by
\(\mathcal{U}_P(P_i, T) \;=\; \max_{n} b_{n,T} \inlineeqnum \label{eq:proposer_utility}\),
which reflects the direct payoff the proposer receives from selecting the highest bid at round \(T\), where the maximum is taken over all \(\{B_i\}\)'s bid.

\section{Mathematical Analysis}
\label{sec:maths}

In this section, we employ a mathematical analysis, using a formalized representation outlined in \autoref{sec:formalization}, to evaluate the performance of the \ac{epbs} system.

Due to the complexity of this system, we only conduct theoretical analysis on some extreme cases to gain insight into the range of results we should be able to observe. Later, we carried out experiments in \autoref{sec:results} with differing parameters to show the dynamics of the system.

\subsection{Bidding Strategy}
\label{sec:bidding_strategy}

For this section, we aim to find an optimal bidding strategy that maximizes a builder's utility, and the auction outcome.

\begin{lemma}[Dominated Bidding Strategies] 

\begin{sublemma}[Passive Participation]
\label{sublemma:passive}
    For any builder $B_i$ with private valuation $v_{i,t}$ at round $t$, any strategy resulting in passive (non-participation) ($b_{i,t} = 0$) is weakly dominated by competitive participation.
\end{sublemma}
\begin{proof}
    Let $b_{-i,t} = \max_{j \neq i} b_{j,t}$. Compare a passive strategy $\sigma_i^{passive}$ that sets $b_{i,t}^{passive} = 0$ with an active strategy $\sigma_i^{active}$ that bids $b_{i,t}^{active} = v_{i,t} - \epsilon$ for arbitrarily small $\epsilon > 0$.
    \begin{enumerate}
        \item If $b_{-i,t} < v_{i,t} - \epsilon$, then whenever the auction ends at round $t$, $\sigma_i^{active}$ wins with utility $\epsilon > 0$, while $\sigma_i^{passive}$ yields $0$. Thus the active strategy is strictly better in this case.
        
        \item If $b_{-i,t} \geq v_{i,t} - \epsilon$, both strategies lose at round $t$, giving $0$. However, in later rounds $z > t$, if the highest competing bid falls below $v_{i,z} - \epsilon$, the active strategy can still win with non-negative payoff, whereas the passive strategy continues to give zero. Thus $\sigma_i^{active}$ weakly dominates across all possible termination rounds.
    \end{enumerate}
    Therefore, passive non-participation is weakly dominated by competitive bidding.
\end{proof}

\begin{sublemma}[Overbidding]
\label{sublemma:overbidding}
    For any builder $B_i$ with private valuation $v_{i,t}$ at round $t$, the strategy $b_{i,t} > v_{i,t}$ is weakly dominated by $b_{i,t} = v_{i,t}$.
\end{sublemma} 

\begin{proof}
    Suppose builder $B_i$ adopts an overbidding strategy $\sigma_i^{over}$ that sets $b_{i,t}^{over} > v_{i,t}$, compared to the valuation strategy $\sigma_i^{val}$ with $b_{i,t}^{val} = v_{i,t}$. Let $b_{-i,t} = \max_{j \neq i} b_{j,t}$ denote the highest competing bid.

    An overbidding builder may justify $\sigma_i^{over}$ by reasoning that the probability of the auction ending at round $t$ is small, and that in some later round $z > t$ they could still lower their bid to a feasible level $b_{i,z} \le v_{i,z}$ and win with positive utility. Two cases would result from this strategy:

    \begin{enumerate}
        \item Early termination: Since the proposer may end the auction at any round, an overbidder risks the auction stopping earlier than anticipated. In that case, their expected utility contribution is $(v_{i,t} - b_{i,t}^{over}) < 0$, while bidding truthfully yields $0$. This strictly lowers expected utility whenever the event $b_{-i,t} < b_{i,t}^{over}$ and $T=t$ has positive probability.
        
        \item Bluffing to deter others: A builder might also hope that a very high bid discourages others from competing aggressively. Yet this is inconsistent with rational play: by Lemma~\ref{sublemma:passive}, passive bidding is also dominated, so competitors will not exit the auction in response to high bids. Instead, overbidding only drives the price level upward, tightening the competition and further reducing the overbidder's probability of winning in later rounds.
    \end{enumerate}

    Therefore, across all possible stopping times $T$, overbidding yields weakly lower expected utility than valuation bidding, with strict disadvantage whenever the auction ends earlier than the overbidder anticipates. Hence overbidding is weakly dominated.
\end{proof}
\end{lemma}

\begin{theorem}[Optimal Bidding Strategies]
\label{thm:optimal_strategy}
    In an auction that can terminate at any round, any optimal bidding strategy for builder $B_i$ with private valuation $v_{i,t}$ exhibits the following properties:
    \begin{enumerate}
        \item $b_{i,t} \in (0, v_{i,t}]$ for all rounds $t$.
        \item Participation timing varies across builders: both continuous participation and delayed competitive entry can constitute weakly dominating strategies depending on individual risk preferences.
        \item Reactive adjustment to observed competitor behavior.
    \end{enumerate}
\end{theorem}
\begin{proof}
    \begin{enumerate}
        \item Follows directly from Lemma~\ref{sublemma:overbidding} and Lemma~\ref{sublemma:passive}.
        
        \item Consider two strategies: continuous participation $\sigma_i^{continuous}$ that maintains competitive bids throughout, and delayed entry $\sigma_i^{delay}$ that enters competitively only in later rounds. Continuous participation eliminates the risk of missing profitable opportunities if the proposer terminates the auction early, since it guarantees that $B_i$ always has a valid competitive bid on record. By contrast, delayed entry relies on the assumption that $\mathbb{P}[T=z]$ is concentrated on later rounds. In this case, withholding bids until a favorable round $z$ allows $B_i$ to enter with a competitive bid $b_{i,z} \leq v_{i,z}$ while avoiding the cost of exposing their valuation too early, which could lead to an increase in bid. The expected utility gain from bidding later is that $B_i$ may win with a lower payment, since competitors had less opportunity to react to its valuation. Both strategies can therefore yield positive expected utility: $\sigma_i^{continuous}$ secures opportunities under high early-stopping probability, while $\sigma_i^{delay}$ may achieve higher surplus when late termination is likely. Neither universally dominates, so both constitute weakly optimal strategies depending on risk preferences and the distribution of $T$.
        
        \item For any strategy that conditions bids on observed competitor actions versus any fixed strategy, the reactive approach can replicate the fixed strategy's performance while potentially improving outcomes when beneficial adjustments are possible, establishing weak dominance.
    \end{enumerate}
\end{proof}

\begin{theorem}[Auction Equilibrium Under Latency]
\label{thm:auction_equilibrium}
Let a PBS auction proceed in discrete, synchronous rounds $t = 1, 2, \dots, T$, where one round corresponds to 0.5 seconds. Let $\mathcal{G} = (\mathcal{Z}, \mathcal{E})$ be the network topology with maximum bid propagation delay $\Delta = \max_{i,j} d(i,j)$, where $d(i,j)$ denotes the latency in rounds between builders $B_i$ and $B_j$.

At each round $t$, builder $B_i$ submits a bid $b_{i,t} \in [0, v_{i,t}]$, where $v_{i,t}$ is the private valuation of builder $B_i$ at round $t$.

Let $i^* = \arg\max_i v_{i,T}$ be the index of the highest-valuation builder at the end of the auction. Define $v_{(1),T} \coloneqq v_{i^*,T}$ and $v_{(2),T} \coloneqq \max_{j \ne i^*} v_{j,T}$ as the highest and second-highest valuations at round $T$, respectively.

If all builders follow rational strategies $\sigma_i(v_{i,t}, H_{i,t})$, then:
\begin{enumerate}
    \item The builder $B_{i^*}$ with valuation $v_{(1),T}$ wins the auction;
    \item The winning bid $b_{i^*,T}$ satisfies $v_{(2),T} \le b_{i^*,T} \le v_{(1),T}$;
    \item Let $\epsilon > 0$ denote an arbitrarily small positive bid increment. Then \\ $\lim_{\Delta \to 0} b_{i^*,T} = v_{(2),T} + \epsilon$.
\end{enumerate}

\end{theorem}
\begin{proof}
We build on the standard English auction result~\cite{Mcafee1987AuctionsBidding}, where the dominant strategy is to bid up to one's valuation. The auction thus ends at the second-highest valuation, with the highest-valuation bidder winning.

We prove each item in order:
\begin{enumerate}
    \item Suppose, for contradiction, that a builder \(B_k \ne B_{i^*}\) with valuation \(v_{k,T} < v_{(1),T}\) wins the auction with bid \(b_{k,T} \leq v_{k,T}\). Since \(T \geq 2\Delta\), builder \(B_{i^*}\) observes \(b_{k,T - \Delta} \in H_{i^*,T}\). Then \(B_{i^*}\) can bid \(b_{i^*,T} = b_{k,T - \Delta} + \epsilon\), with \(\epsilon > 0\) arbitrarily small, while satisfying \(b_{i^*,T} \leq v_{(1),T}\), contradicting the assumption that \(B_k\) wins. Therefore, \(B_{i^*}\) must win.
    
    \item By rationality, \(b_{i^*,T} \leq v_{(1),T}\). Moreover, observing the second-highest valuation \(v_{(2),T}\), builder \(B_{i^*}\) bids at least \(b_{i^*,T} \geq v_{(2),T} + \epsilon\), where \(\epsilon > 0\) is an arbitrarily small increment. It follows that
    \(
    v_{(2),T} \leq b_{i^*,T} \leq v_{(1),T}.
    \)
    
    \item As \(\Delta \to 0\), bid propagation becomes instantaneous, and as \(\epsilon \to 0\), the increment above \(v_{(2),T}\) vanishes. Therefore,
    \(
    \lim_{\epsilon \to 0, \Delta \to 0} b_{i^*,T} = v_{(2),T}.
    \)

\end{enumerate}
\end{proof}

In reality, the existence of latency means that the highest-valuation builder does not always win, as delayed bid propagation can prevent timely responses to underbids. We further investigate the impact of such latency effects on the auction outcome through simulation in \autoref{sec:bid_sim}.

\subsection{MEV vs. Non-MEV Builders}

Assume two builders, \(B_1\) ($\mathtt{attack}$ ) and \(B_2\) ($\mathtt{benign}$), with access to identical mempools, i.e. \(\mathcal{M}_{B_1,t} = \mathcal{M}_{B_2,t}\), and employing the same bidding strategy \(\sigma_1 = \sigma_2\), later in \autoref{sec:simulation} the difference between reactive bidding strategies and last minute bidding is discussed in more detail. The distinction lies in their transaction selection behavior: builder \(B_1\) includes attack transactions \(x_k \in \mathbf{x}_{B_1,t}\) that target other transactions via the mapping \(\mathcal{A}(x_k) = x_j\) and yield \ac{mev} profits \(m_j > 0\), while builder \(B_2\) refrains from such behavior and includes no such transactions, i.e., \(\mathbf{x}_{B_2,t} = \emptyset\). 

\begin{theorem}[\ac{mev} Advantage]
\label{thm:centralization}
Let \(B_1\) be a $\mathtt{attack}$ builder and \(B_2\) a $\mathtt{benign}$ builder with identical mempools \(\mathcal{M}_{B_1,t} = \mathcal{M}_{B_2,t}\) and identical bidding strategies \(\sigma_1 = \sigma_2\). Then the private valuation of the $\mathtt{attack}$ builder is weakly greater:
\(v_{1,t} \geq v_{2,t}, \quad \forall t.\)

\end{theorem}
\begin{proof}
Since both builders observe the same mempool and use the same bidding strategy \(\sigma(v_{i,t}, H_t)\), the difference in bids arises solely from the difference in private valuations.

By definition of the valuation function in \autoref{eq:builder_valuation}, \(v_{1,t}\) contains an additional non-negative term:
\(\sum_{x_k \in (\mathcal{T}_{1,t} \cap \mathbf{x}_{B_1,t})}
\mathbf{1}_{x_j = \mathcal{A}(x_k)} \cdot m_j \geq 0,
\)
representing MEV gains from attacking transactions initiated by \(B_1\), which \(B_2\) does not include.

Due to block size constraints, \(\mathcal{T}_{1,t}\) and \(\mathcal{T}_{2,t}\) may differ due to strategic inclusion decisions. However, a rational builder \(B_1\) includes \(x_k\) only if doing so increases block value, that is, if
\((\sum_{x_k \in (\mathcal{T}_{1,t} \cap \mathbf{x}_{B_1,t})} \mathbf{1}_{x_j = \mathcal{A}(x_k)} \cdot (m_j + g_j)
\;>\;
\sum_{x_r\in \mathcal{T}_{2,t}\setminus \mathcal{T}_{1,t}} g_r,
\)
where \(\mathcal{T}_{2,t} \setminus \mathcal{T}_{1,t}\) denotes transactions included by \(B_2\) but excluded by \(B_1\) to make space for the attacking transactions.

This decision is applied iteratively over all possible attacks, ensuring that the block valuation \(v_{1,t}\) is at least as large as \(v_{2,t}\), hence\(v_{1,t} \geq v_{2,t}.\)
Strict inequality holds if at least one attacking transaction generates strictly positive net value.
\end{proof}

Although \autoref{thm:centralization} concerns only the valuation gap between 
$\mathtt{attack}$ and $\mathtt{benign}$ builders, its implications differ fundamentally between PoS and ePBS. Under standard PoS, block proposers are selected uniformly at random (proportional to stake), and private valuation does not influence block production probability. Thus, an MEV-seeking validator may increase its per-block returns but does not increase its probability of producing blocks.

Combining the results of \autoref{thm:auction_equilibrium} and \autoref{thm:centralization}, we provide a partial theoretical answer to \textbf{RQ1.1} and \textbf{RQ1.2}. Specifically, under the assumptions of identical mempools and symmetric bidding strategies, $\mathtt{attack}$ builders are shown to consistently achieve higher private valuations for their proposed blocks, hence winning the auction by submitting the highest bid. Consequently, $\mathtt{attack}$ builders are theoretically guaranteed to win all auctions in this simplified setting. A more general analysis in a multi-agent, heterogeneous environment is conducted in \autoref{sec:rq1}.

\subsection{Reinvesting Staking}
\label{sec:reinvest}
In \ac{pos} without \ac{epbs}, centralization emerges naturally through restaking: participants who earn higher rewards accumulate more stake over time, increasing their probability of being selected again, an effect often summarized as \enquote{the rich get richer} \cite{huang2021}. While the previous analysis focused on single-block strategic interactions, we now analyze long-term reinvestment behavior in \ac{epbs} with vertical integration of builders and proposers \cite{Leonardo2022}. 

Building on the single-block formalization in \autoref{sec:formalization}, we now extend to multi-block dynamics. Let $s_i(\ell)$ denote the stake of node $i$ at slot $\ell$, and $\gamma_i \in \{0,1\}$ the reinvestment factor of node $i$ be the multiplier of how much profit is reinvested for staking. 
The probability of node $i$ being selected to be the block confirming participant at slot $\ell$ is:
$\frac{s_i(\ell)}{\sum_{j \in B_i, P_i} s_j(\ell)},$ 
where the chance of being selected is a uniform probability based on the stake of a single participant over all stakes of all participants, which includes both the proposers and builders.

\paragraph{Builder}
The stake evolution for builder $B_i$ follows
{\scriptsize
\[
\mathbb{E}[s_{B_i}(\ell+1)]
=
s_{B_i}(\ell)
+
\gamma_{B_i}
\Bigg[
\underbrace{
\frac{s_{B_i}(\ell)}{\sum_{j \in B_i, P_i} s_j(\ell)}
}_{\substack{\text{Probability of being}\\\text{selected as proposer}}}
\overbrace{
v_{i,T}(\ell)
}^{\substack{\text{Revenue}\\\text{as proposer}}}
+
\underbrace{
\left(
1 -
\frac{s_{B_i}(\ell)}{\sum_{j \in B_i, P_i} s_j(\ell)}
\right)
}_{\substack{\text{Probability of not}\\\text{being selected as proposer}}}
\overbrace{
U_B(B_i,T,\ell)
}^{\substack{\text{Auction revenue}\\\text{as builder}}}
\Bigg],
\]
}
where
\[
U_B(B_i,T,\ell)
=
\overbrace{
p_{B_i}^{win}(\ell)
}^{\substack{\text{Probability of winning}\\\text{the auction}}}
\,
\overbrace{
\bigl(
v_{i,T}(\ell) - b_{i,T}(\ell)
\bigr)
}^{\substack{\text{Revenue gained}\\\text{from auction}}}.
\]

The growth rate is:
\[\frac{s_{B_i}(\ell+1)}{s_{B_i}(\ell)} = 1 + \frac{\gamma_{B_i} R_{B_i}(\ell)}{s_{B_i}(\ell)} = 1 + \gamma_{B_i} \left[\frac{v_{i,T}(1-f \pi)}{\sum_j s_j} + \frac{f \pi v_{i,T}}{s_{B_i}}\right],\]
where $f$ is a function of valuation for the probability of winning the auction, and $\pi \in [0,1]$ is the fraction of valuation builder will gain.

Substituting in and simplifying (we are going to drop $\ell$):
$$\frac{s_{B_i}(\ell+1)}{s_{B_i}(\ell)} = 1 + \gamma_{B_i} \left[\frac{s_{B_i}}{\sum_j s_j} \cdot v_{i,T} + \left(1 - \frac{s_{B_i}}{\sum_j s_j}\right) U_B\right] \frac{1}{s_{B_i}}$$
$$\frac{s_{B_i}(\ell+1)}{s_{B_i}(\ell)} = 1 + \gamma_{B_i} \left[\frac{v_{i,T}}{\sum_j s_j} + \frac{U_B}{s_{B_i}} - \frac{U_B}{\sum_j s_j}\right].$$
For $U_B(B_i,T,\ell) = p_{B_i}^{win}(\ell)(v_{i,T}(\ell) - b_{i,T}(\ell))$, let the probability of winning the auction be a function of valuation and auction abilities taking value $[0,1]$. Let $p_{B_i}^{win}(\ell)$ be a function depending on valuation $f(v)$, and since $b < v$, replace $(v_{i,T}(\ell) - b_{i,T}(\ell))$ with $\pi v_{i,T}(\ell)$ where $\pi \in [0,1]$. This is the profit margin of the builder.
Substituting back, we get:
$$\frac{s_{B_i}(\ell+1)}{s_{B_i}(\ell)} = 1 + \gamma_{B_i} \left[\frac{v_{i,T}}{\sum_j s_j} + \frac{f \pi v_{i,T}}{s_{B_i}} - \frac{f \pi v_{i,T}}{\sum_j s_j}\right]$$
$$\frac{s_{B_i}(\ell+1)}{s_{B_i}(\ell)} = 1 + \gamma_{B_i} \left[\frac{v_{i,T}(1-f \pi)}{\sum_j s_j} + \frac{f \pi v_{i,T}}{s_{B_i}}\right].$$

\begin{figure}[h!]
    \centering
    \begin{minipage}[t]{0.45\textwidth}
        \centering
        \includegraphics[width=\linewidth]{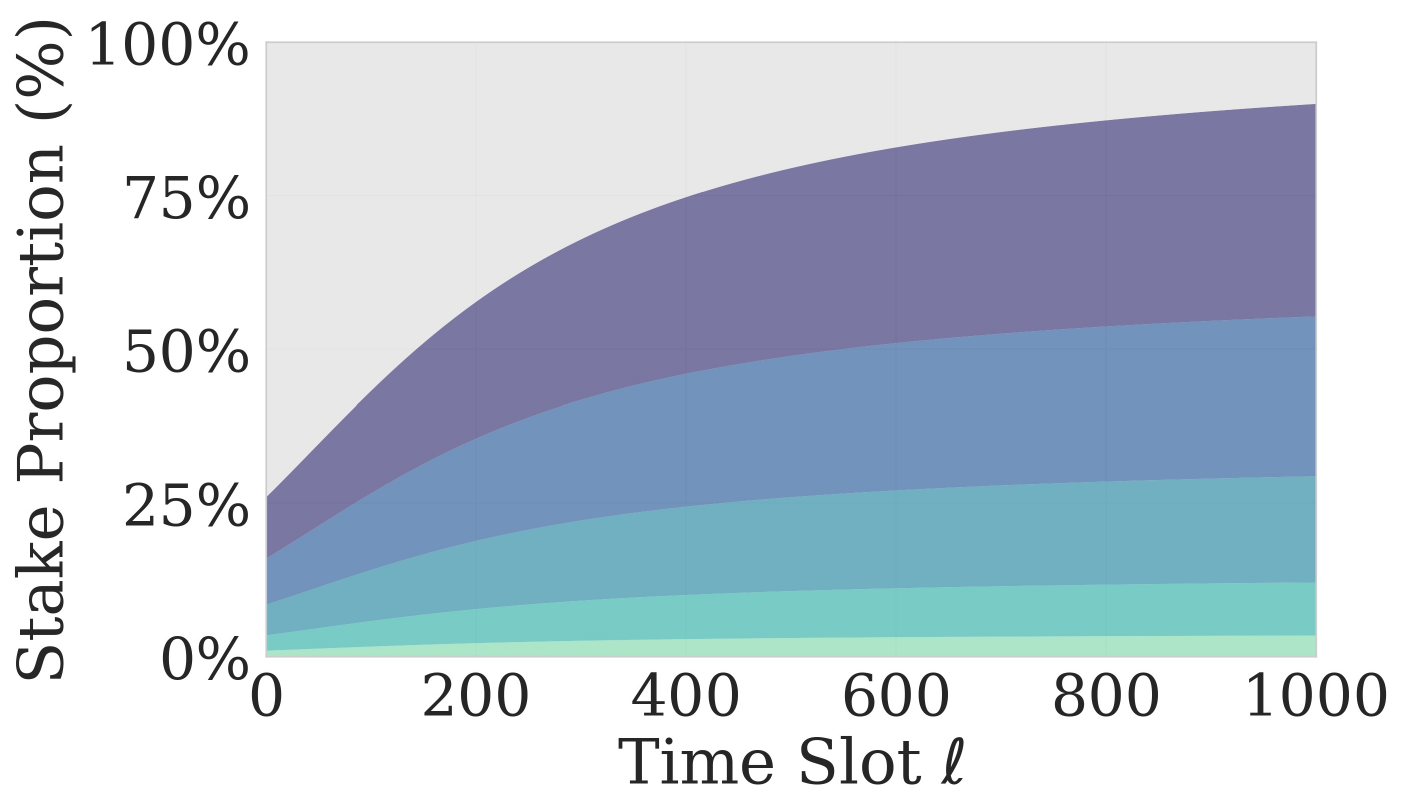}  
        \caption{Builder's stake over slots.}
        \label{fig:theory_stake}
    \end{minipage}
    \begin{minipage}[t]{0.45\textwidth}
        \centering
        \includegraphics[width=\linewidth]{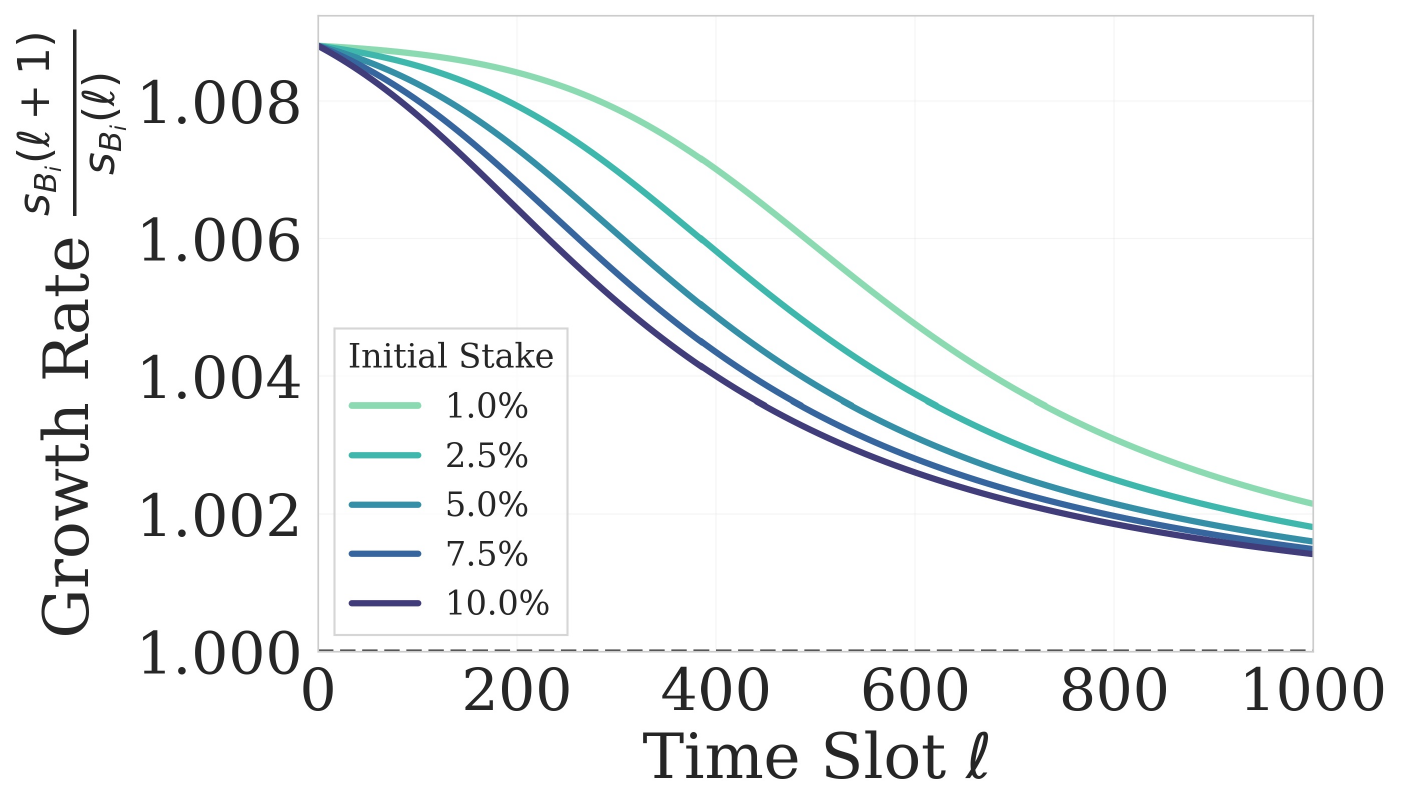}
        \caption{Builder's growth rate compared to current stake level.}
        \label{fig:theory_growth_rate}
    \end{minipage}
\end{figure}

\paragraph{Proposer}
The stake evolution for proposer $P_i$ is:
\[
s_{P_i}(\ell+1)
=
s_{P_i}(\ell)
+
\gamma_{P_i}\,
\frac{s_{P_i}(\ell)}{\sum_j s_j(\ell)}\,
b_{i,T}(\ell).
\]

The growth rate is:
\[
\frac{s_{P_i}(\ell+1)}{s_{P_i}(\ell)}
=
1
+
\gamma_{P_i}\,
\frac{b_{i,T}(\ell)}{\sum_j s_j(\ell)}.
\]

\paragraph{Validator}
The stake evolution for validator $V_i$ is:
\[
s_{V_i}(\ell+1)
=
s_{V_i}(\ell)
+
\gamma_{V_i}\,
\frac{s_{V_i}(\ell)}{\sum_j s_j(\ell)}\,
v_{i,T}(\ell).
\]

The growth rate is:
\[
\frac{s_{V_i}(\ell+1)}{s_{V_i}(\ell)}
=
1
+
\gamma_{V_i}\,
\frac{v_{i,T}(\ell)}{\sum_j s_j(\ell)}.
\]

From these growth functions, all agents exhibit positive stake accumulation. Validators centralize faster than proposers when $v > b$, since their growth rate scales directly with valuation. Builders exhibit an additional term $\tfrac{f \pi v_{i,T}}{s_{B_i}}$, which is large when $s_{B_i} \ll \sum_j s_j$ but diminishes as $s_{B_i}$ grows. Thus, their relative growth rate slows with higher stake, though absolute increments $\Delta s_{B_i} = \gamma_{B_i} R_{B_i}(\ell)$ may still increase due to the larger stake, as shown in \autoref{fig:theory_stake} and \autoref{fig:theory_growth_rate}. 

This analysis answers \textbf{RQ1}: in \ac{pos} with \ac{epbs}, the \enquote{rich get richer} effect \cite{huang2021} is amplified by builder–proposer integration. Builders accumulate capital even when not selected as proposers, reinforcing long-term centralization of block content. We further examine parameter sensitivities in \autoref{sec:restaking}.

\subsection{Expected MEV Inclusion Probability Under ePBS and PoS}
\label{sec:pbs_pos_mev_expectation}

We analyze the probability that the winning block under \ac{pos} with and without \ac{epbs} contains at least one profitable \ac{mev} transaction pair. 

Let \(\mathcal{T}_{i,T}\) be the ordered list of transactions in the block proposed by builder \(B_i\) at round \(T\). Define the event  
\(E_{\mathrm{MEV}} \coloneqq \{\exists\, x_j, x_k \in \mathcal{T}_{i,T} : \mathcal{A}(x_k) = x_j,\; m_j > 0\}\),  
that is, the block includes at least one profitable \ac{mev} pair.  

Let \(\theta = \mathbb{P}_{\mathrm{system}}[E_{\mathrm{MEV}}]\) denote the baseline prior probability. Introduce \(W\) as the event that a builder's block wins the auction. The posterior inclusion probabilities are  
\(\phi_{\mathrm{ePBS}} \coloneqq \mathbb{P}_{\mathrm{ePBS}}[E_{\mathrm{MEV}} \mid W]\) and \(\phi_{\mathrm{PoS}} \coloneqq \mathbb{P}_{\mathrm{PoS}}[E_{\mathrm{MEV}} \mid W]\).

\begin{proposition}[Block Producer–Initiated MEV]
\label{prop:pbs_mev_bias}
Under \ac{epbs}, the probability of including \ac{mev} transaction pairs is weakly higher than under \ac{pos}.
\end{proposition}
\begin{proof}
    Let \(\omega = \frac{\mathbb{P}_{\mathrm{ePBS}}[W \mid E_{\mathrm{MEV}}]}{\mathbb{P}_{\mathrm{ePBS}}[W \mid \neg E_{\mathrm{MEV}}]} \geq 1\) representing the advantage of winning the auction if the block has included \ac{mev} transactions derived from \autoref{thm:centralization}.  
Set \(a = \mathbb{P}_{\mathrm{ePBS}}[W \mid \neg E_{\mathrm{MEV}}]\), then \(\mathbb{P}_{\mathrm{ePBS}}[W \mid E_{\mathrm{MEV}}] = \omega a\).  
By the law of total probability, \(a(1-\theta) + \omega a \theta = 1\), hence \(a = \frac{1}{1 + \theta(\omega - 1)}\).  
Applying Bayes' rule yields  
\(\phi_{\mathrm{ePBS}} = \frac{\omega \theta}{1 + \theta(\omega - 1)}\).

\end{proof}

\begin{proposition}[User–Initiated MEV]
\label{prop:user_mev_invariance}
When block producers (builders and validators) are $\mathtt{benign}$ and proposers do not vertically integrate to be builders, the inclusion probability of user-initiated \ac{mev} pairs is equal across \ac{epbs} and \ac{pos}.
\end{proposition}
\begin{proof}
    Since block selection under \ac{pos} is independent of \ac{mev} content, we obtain  
\(\phi_{\mathrm{PoS}} = \theta\).
\end{proof}

\begin{theorem}[MEV Inclusion Probability]
\label{thm:pbs_pos_mev}
Taken together, proposition \ref{prop:pbs_mev_bias} and \ref{prop:user_mev_invariance} imply that \ac{epbs} strictly increases the posterior inclusion probability of \ac{mev} pairs when $\mathtt{attack}$ producers exploit opportunities, but reduces to the baseline prior when producers are $\mathtt{benign}$.
\end{theorem}

\section{Simulation Design}
\label{sec:simulation}

This section outlines the simulation framework developed to analyze the behavior of various participants in a blockchain environment under different conditions. Our use of agent-based simulation modeling follows established practice for studying strategic interaction in auction-based and blockchain systems \cite{Carbonneau2016,Jeong2026LiqBoost:Exchanges,Ntuala2026ADetection,Ranganatha2025}.
Users, builders, proposers, and network latency are modeled according to the setup in \autoref{sec:formalization}. The simulation code is available at \href{https://github.com/xujiahuayz/mev-resist-consensus}{our GitHub repository}.

\subsection{Assumptions and Parameters}
\label{sec:sim_assump}

The simulation is based on several key parameters and assumptions: 
For the \ac{epbs} system, it includes 100 users, 50 builders, and 50 proposers; for the \ac{pos} system, it includes 100 users and 50 validators, testing with the number of \ac{mev}-seeking participants from none to all. We choose 50 validators as a parameter because we compare the validators with proposer and builders in \ac{epbs} separately, used as a benchmark. The block capacity is 100, and the gas fee is randomly sampled, where both are based on the information fetched on-chain from Ethereum mainnet data. We sample multiple periods after the merge in 2022, including stable periods and highly volatile periods, for example, the USDC depeg, FTX crash, and Terra Luna fall.
\ac{mev} are sampled from the distribution given in FlashBot data \cite{mev_distribution} which shows a Gamma distribution. The simulation primarily focuses on simple \ac{mev} strategies such as front running, back running, and sandwich attacks. Latency and network connectivity follow the design in \autoref{sec:formalization}, the specific graph we choose for the simulation is the Erdos-Renyi random graph, which is often used as a baseline to analyze the Ethereum network topology~\cite{Li2021TopoShot:Transactions}. All other parameters are kept the same throughout the simulation for consistency. Each simulation run includes 1,000 blocks. 

For analyzing long-term centralization effects, we extend the simulation with additional parameters modeling stake accumulation dynamics. Each consensus participant is assigned a reinvestment factor $\gamma_i = \{0,1\}$ with equal probability, reflecting the polarized behavior observed in practice where validators either extract all profits ($\gamma_i = 0$) or participate in liquid staking protocols with automatic compounding ($\gamma_i = 1$) \cite{Beccuti2025TowardsMarket}. Initial stake distributions are varied across participants to simulate realistic heterogeneity, with total network stake normalized to represent approximately 1,000 active validators. The discrete staking threshold is set to $32 \times 10^9$ gwei (equivalent to 32 ETH), reflecting Ethereum's validator requirements. Participants accumulate capital continuously through $k_i(\ell+1) = k_i(\ell) + \gamma_i \cdot R_i(\ell)$, but active stake updates only when crossing threshold boundaries: $s_i(\ell+1) = (32 \times 10^9) \cdot \lfloor k_i(\ell+1) / (32 \times 10^9) \rfloor$. For these long-term equilibrium studies, simulations are extended to 10,000 blocks to observe convergence behavior.

\subsection{Simulation Process}

In the simulation, we model both \ac{pos} with \ac{epbs} and without \ac{epbs} to compare their performance under similar parameters.

In the \ac{epbs} system, builders select transactions for inclusion in their blocks based on their strategy, whether $\mathtt{benign}$ or $\mathtt{attack}$. After selecting transactions, builders calculate and submit bids to propose their blocks according to their strategy. In our simulation, we use two strategies: $\mathtt{reactive}$ and $\mathtt{last} \space \mathtt{minute}$, which are among the most used in current relay-based \ac{pbs} deployments~\cite{ThomasThiery2023EmpiricalBBPs, Wu2023}. A reactive builder at round~\(t\) bids 
\(\, b_{i,t} = \min\{ v_{i,t},\; \max_{j \in \{B_i\}} b_{j,k} + \delta \mid k \le t \}\,\),
matching the highest observed bid in history with a small increment \(\delta > 0\), constrained by its private valuation \(v_{i,t}\). A last-minute builder, by contrast, waits until round~\( t \ge c \) where \(c\) is a threshold round near the end of the auction, and then bids 
\(\, b_{i,t} = \min\{ v_{i,t},\; \max_{j \in \{B_i\}} b_{j,k} + \delta \mid k \le t \} \cdot \mathbf{1}_{\{t \ge c\}} \,\),
using the same competitive bidding strategy as reactive builders but only participating from their chosen threshold round \(c\) onward. This strategy allows the builder to delay participation and reduce reaction times for other builders; however, they risk not having submitted a bid at the end of the auction or that their bid has not propagated to the proposer.

The highest bid at the end of the auction is selected, and the corresponding block is included in the blockchain. All proposers deploy an adaptive auction termination strategy, where the auction stopping round \(T_s \in \{1,2,\dots,24\}\) for slot~\(s\) is adjusted based on observations from the previous slot~\(s{-}1\). Specifically, the proposer observes whether a higher bid than the winning bid \(b_w\) of slot \(s{-}1\) was submitted after or before the auction ended in that slot. The strategy is defined as
\[
T_s =
\begin{cases}
T_{s-1} + 1, & \text{if } \exists\, t > T_{s-1} \text{ s.t. } b_{i,t} > b_w, \\
T_{s-1} - 1, & \text{if } \exists\, t < T_{s-1} \text{ s.t. } b_{i,t} > b_w, \\
T_{s-1},     & \text{otherwise}.
\end{cases}
\]
Here, $T_s$ is constrained such that $1 \leq  T_s \leq 24$, respecting the maximum number of auction rounds allowed per slot. If the auction is terminated at round \(T_s\), any bids \(b_{i,t}\) with \(t > T_s\) are ignored for that slot, though they may inform proposer decisions in subsequent slots. This adaptive mechanism allows proposers to reward early bidding and discourage excessive last-minute bidding, while avoiding timeout risk by enforcing a strict upper bound.   

In the \ac{pos} system, a validator is randomly selected to propose a block. The validator selects transactions from its mempool \(\mathcal{M}_{V_i,t}\) at any time \(t \in \{1, 2 \dots, 24\}\) according to its attribute $\tau_{V_i}$, and the resulting block is validated and added to the chain.

\section{Simulation Results}
\label{sec:results}

In this section, we present the findings of our study on the impact of \ac{epbs} on decentralization and its economic implications, based on simulations calibrated with Ethereum on-chain data. We analyze the distribution of builders and validators, the effects on transaction dynamics, and compare the outcomes of \ac{epbs} with those of \ac{pos}. Each subsection addresses specific research questions through targeted experiments and analyses.

\subsection{RQ1 Answer: Impact of \ac{epbs} on Decentralization}
\label{sec:rq1}

The hypothesis is that \ac{epbs} will result in a broader distribution of profits within participants, indicating increased decentralization, as it aims to reduce the influence of dominant participants and promote a more balanced and equitable network ecosystem. To test this hypothesis, we conduct experiments to measure the decentralization of \ac{epbs}. We compare the profit and block-building dynamics between \ac{epbs} and \ac{pos}.

\subsubsection{Builder and Validator Distribution}
To assess decentralization, we first analyze the distribution of builders and validators and investigate how the distribution of builders changes with different percentages of \ac{mev} builders and \ac{mev} users.

\begin{figure}[tb]
    \centering
    \begin{subfigure}[t]{0.45\textwidth}
        \centering
        \includegraphics[height=0.25\textheight]{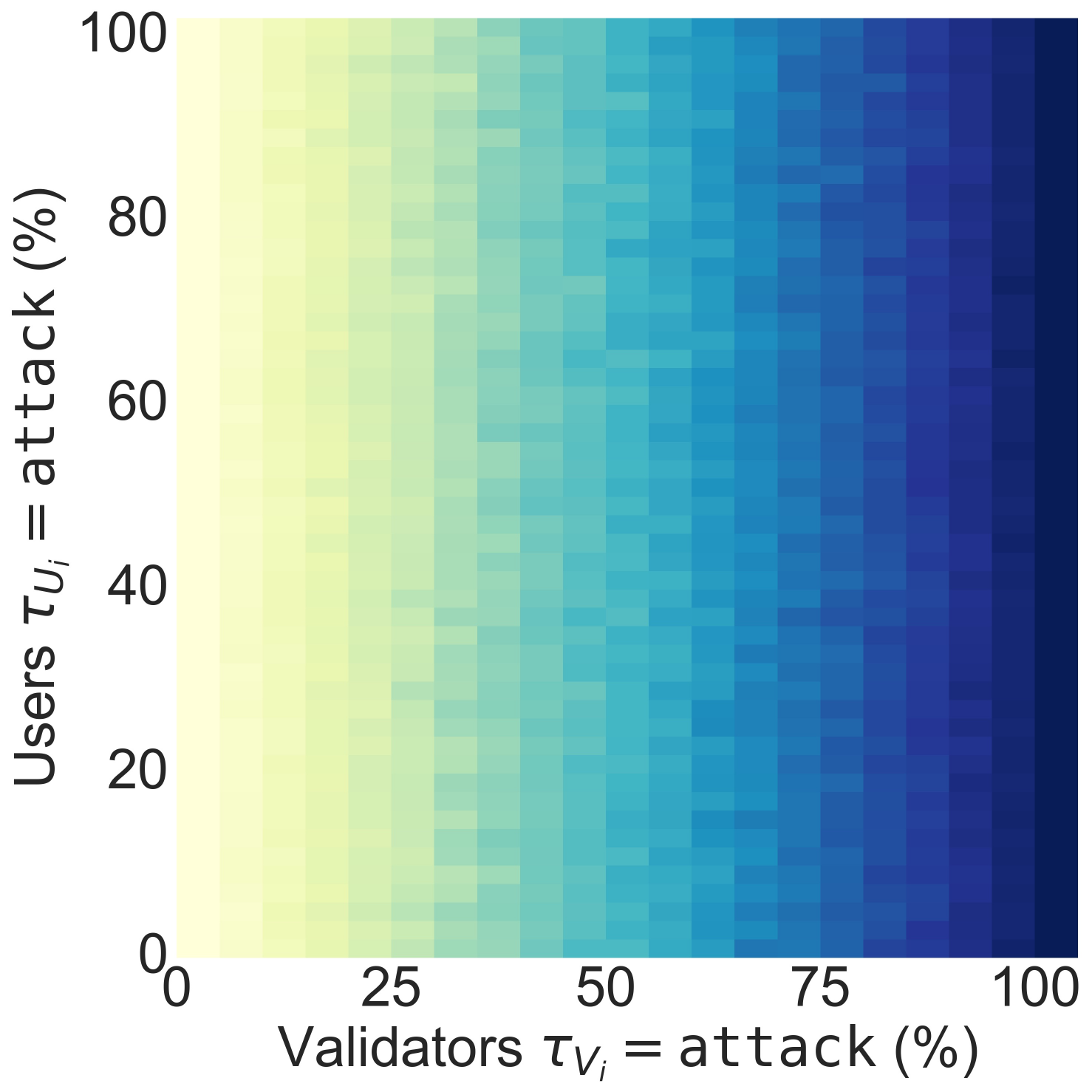}
        \caption{PoS.}
        \label{fig:pos_cumu}
    \end{subfigure}\hfill
    \begin{subfigure}[t]{0.45\textwidth}
        \centering
        \includegraphics[height=0.25\textheight]{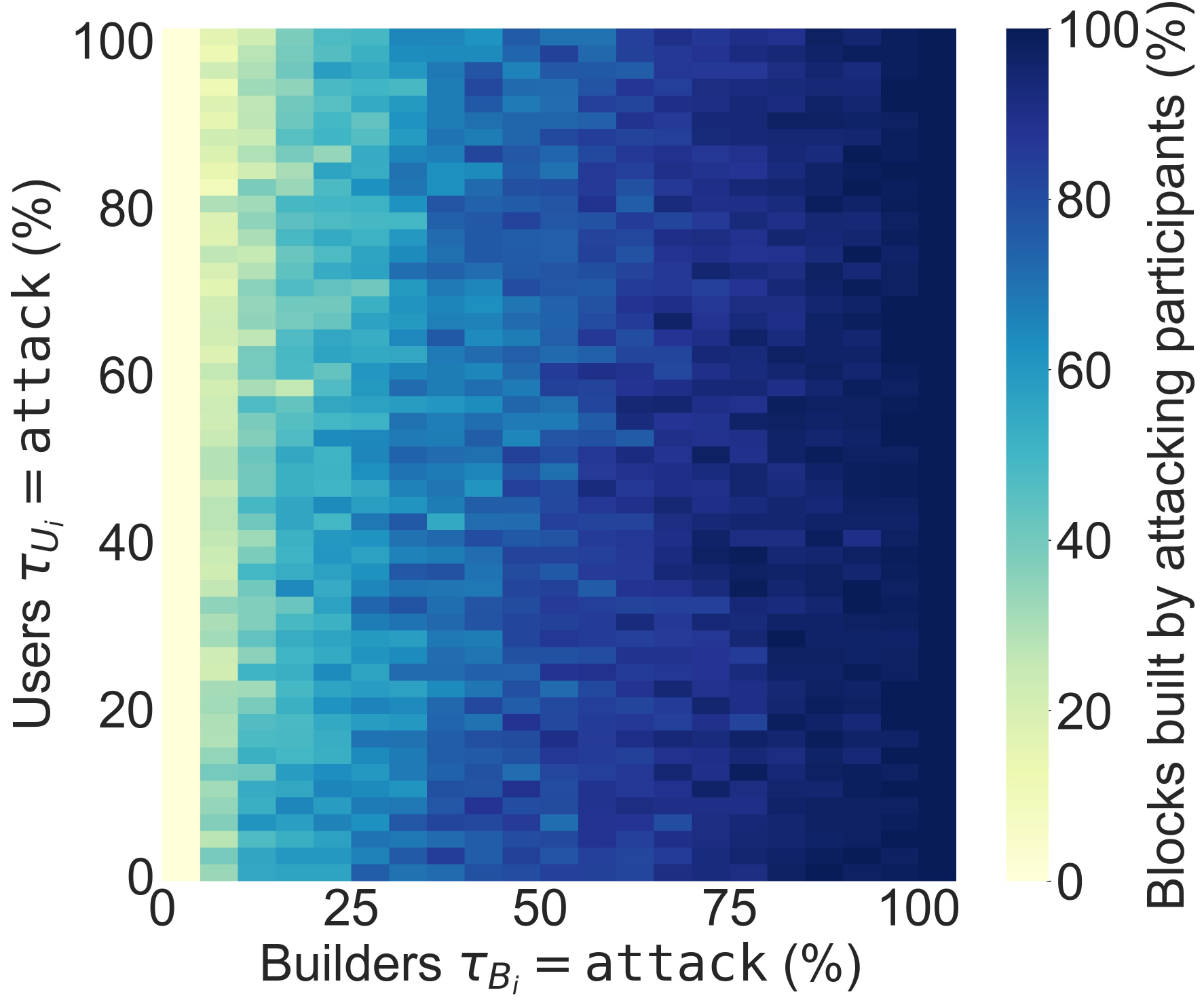}
        \caption{ePBS.}
        \label{fig:pbs_cumu}
    \end{subfigure}
    \caption{Comparison of block production by attacking and benign participants.}
    \label{fig:side_by_side}
\end{figure}

From \autoref{fig:pbs_cumu} and \autoref{fig:pos_cumu}, we observe clear differences in block production patterns. Under \ac{epbs}, the share of blocks produced by \ac{mev}-seeking builders increases steadily with their presence in the system, reflecting their competitive advantage and consistent dominance. In contrast, under \ac{pos}, block production remains evenly distributed across validators, consistent with uniform random selection.  

The figures also show that the proportion of user attackers has little effect on which participant produces the final block, confirming proposition~\ref{prop:user_mev_invariance} that user-initiated transactions do not influence auction outcomes. Overall, \ac{epbs} amplifies centralization among \ac{mev}-seeking builders, while \ac{pos} preserves a more balanced distribution. This aligns with our theoretical results in \autoref{thm:centralization} and answers \textbf{RQ1.1}.

\subsubsection{Profit Distribution between PoS and ePBS}
\label{sec:profit_dis}

\begin{figure}[t]
    \centering

    \begin{subfigure}[b]{0.9\textwidth}
        \centering
            \includegraphics[width=\linewidth]{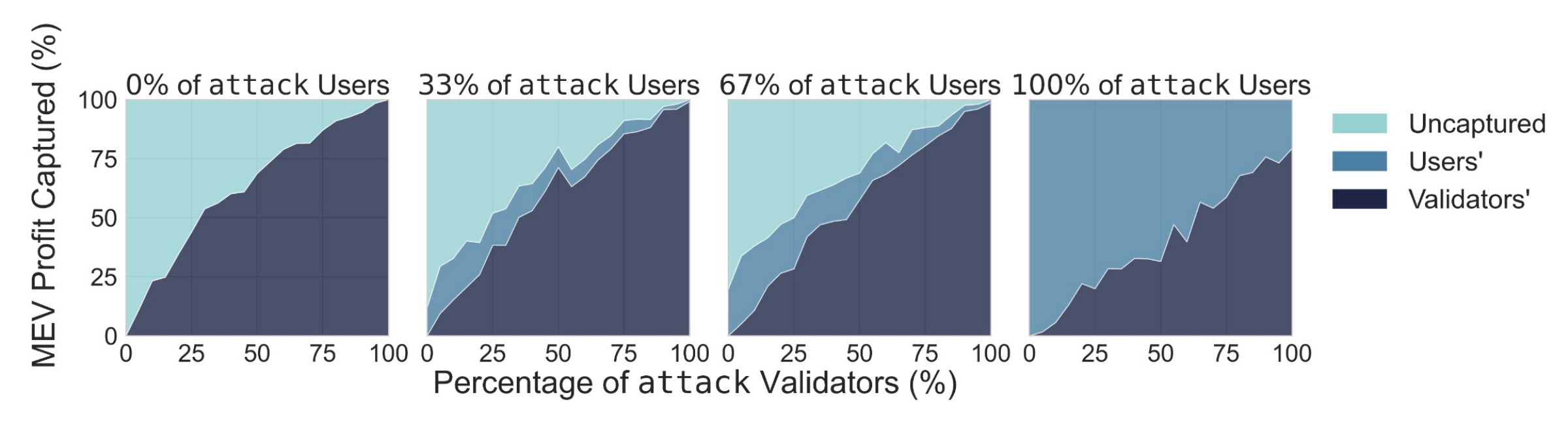}
        \caption{PoS.}
        \label{fig:pos_mev_reward_dis}
    \end{subfigure}

    \vspace{1em}

    \begin{subfigure}[b]{0.9\textwidth}
        \centering
            \includegraphics[width=\linewidth]{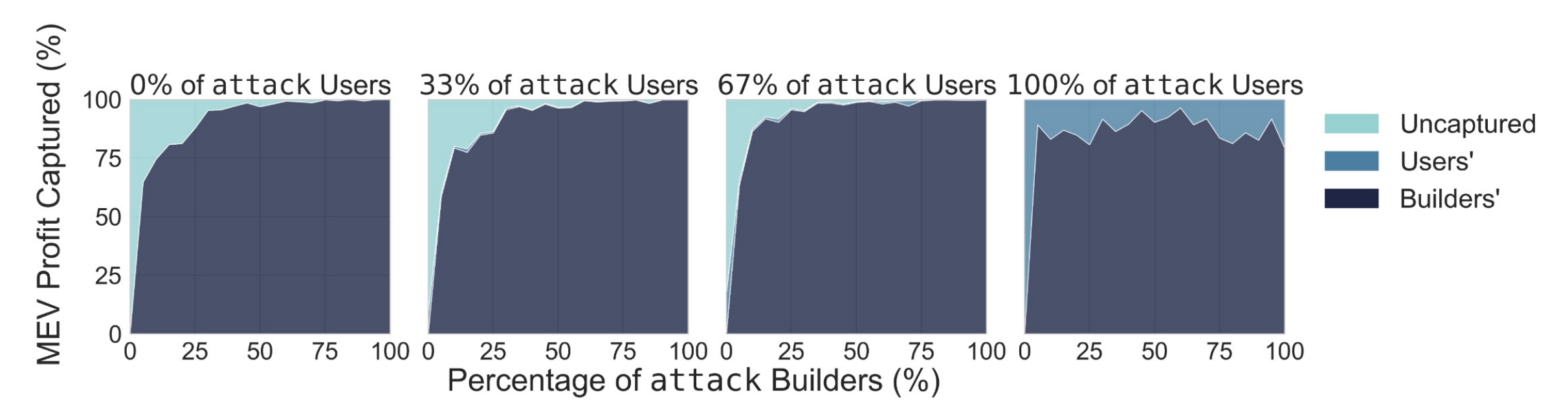}
        \caption{ePBS.}
        \label{fig:mev_reward_dis}
    \end{subfigure}

    \caption{Comparison of the percentage of MEV profit distribution among users, validators, and uncaptured MEV under varying numbers of attacking users. }
    \label{fig:mev_distribution_side_by_side}
\end{figure}

\autoref{fig:mev_reward_dis} and \autoref{fig:pos_mev_reward_dis} show the distribution of realized \ac{mev} profits among users and block producers. User-extracted \ac{mev} is given by the sum over \( m_k \) where \( x_k \in \mathcal{T}_{w,T} \) and \( \exists\, x_j \in \mathbf{x}_{U_i,T} \) such that \( x_k = \mathcal{A}(x_j) \). Builder- or validator-extracted \ac{mev} corresponds to \( x_j \in \mathbf{x}_{U_i,T} \) or \( x_j \in \mathbf{x}_{V_i,T} \), matching the second term in \autoref{eq:builder_valuation}. Note that the plot excludes the proposer's share of builder profits, as proposers typically capture nearly the entire block value. Including it would misleadingly suggest that all builder profits are attributed to proposers.

In \ac{epbs} (\autoref{fig:mev_reward_dis}), most \ac{mev} is captured by builders, with builder-initiated attacks consistently outweighing user-initiated ones. This dominance persists across all levels of user attacks. By contrast, under \ac{pos} (\autoref{fig:pos_mev_reward_dis}), user-extracted \ac{mev} rises sharply as more users attack, while validator share declines.  

These results answer \textbf{RQ1.2}: \ac{epbs} disproportionately rewards builders, reinforcing their dominance, whereas \ac{pos} yields a more balanced division of \ac{mev} between users and validators.

\subsubsection{Bidding Dynamics Between Builders}
\label{sec:bid_sim}

Having established stronger centralization of profit and content under \ac{epbs}, we now examine the underlying driver: the block auction.  

In this setting, half of the users submit attack transactions, half of the builders are \ac{mev}-seeking (\(\tau_{B_i} = \mathtt{attack}\)), and bidding strategies are split: one-quarter of builders adopt a last-minute strategy, while three-quarters follow a reactive strategy.

\begin{figure}[ht!]
    \centering
    \begin{minipage}[t]{0.4\textwidth}
        \centering
        \includegraphics[width=\linewidth]{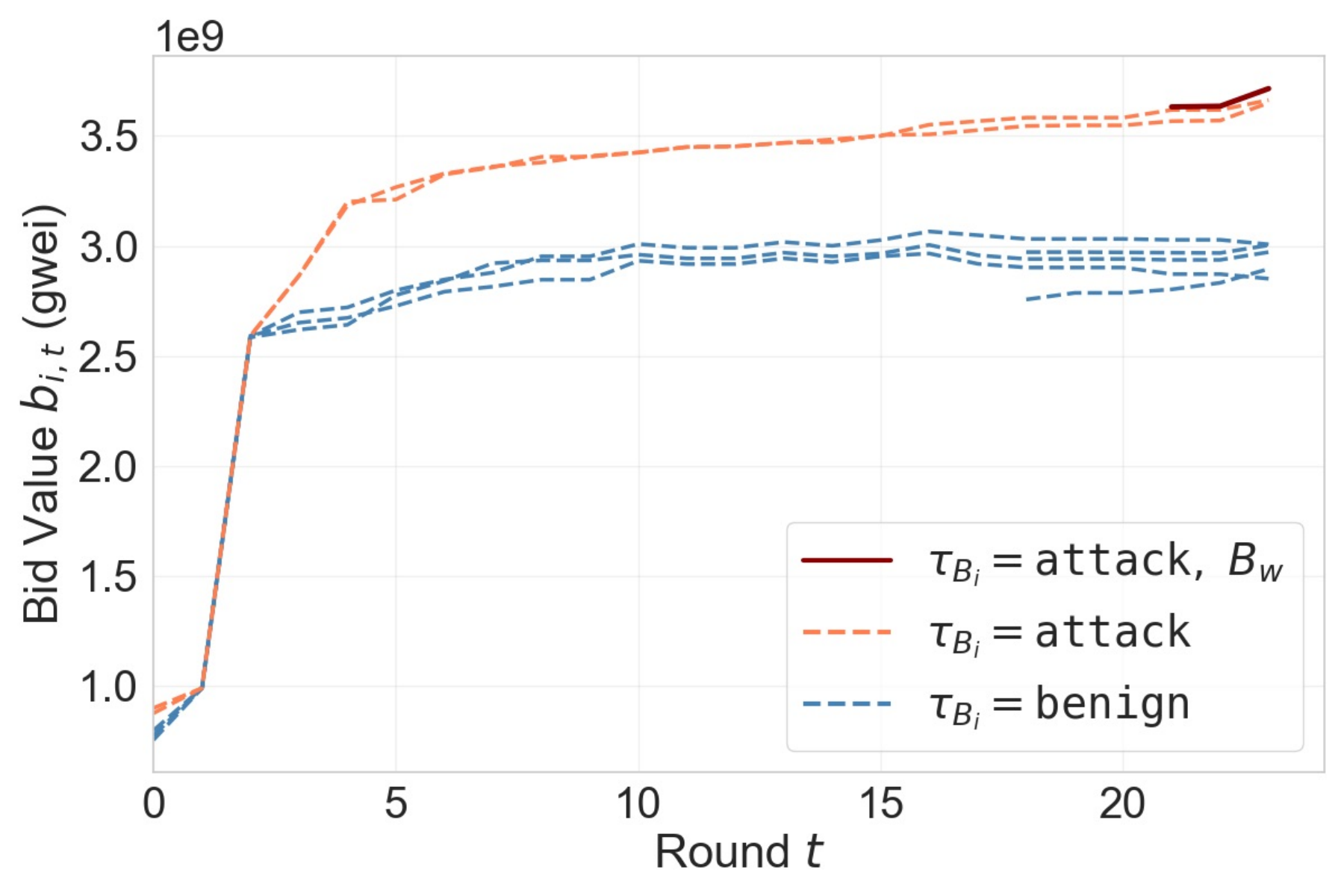}
        \caption{Auction bids over time.}
        \label{fig:bid_dynamic}
    \end{minipage}
    \hspace{1em} 
    \begin{minipage}[t]{0.4\textwidth}
        \centering
        \includegraphics[width=\linewidth]{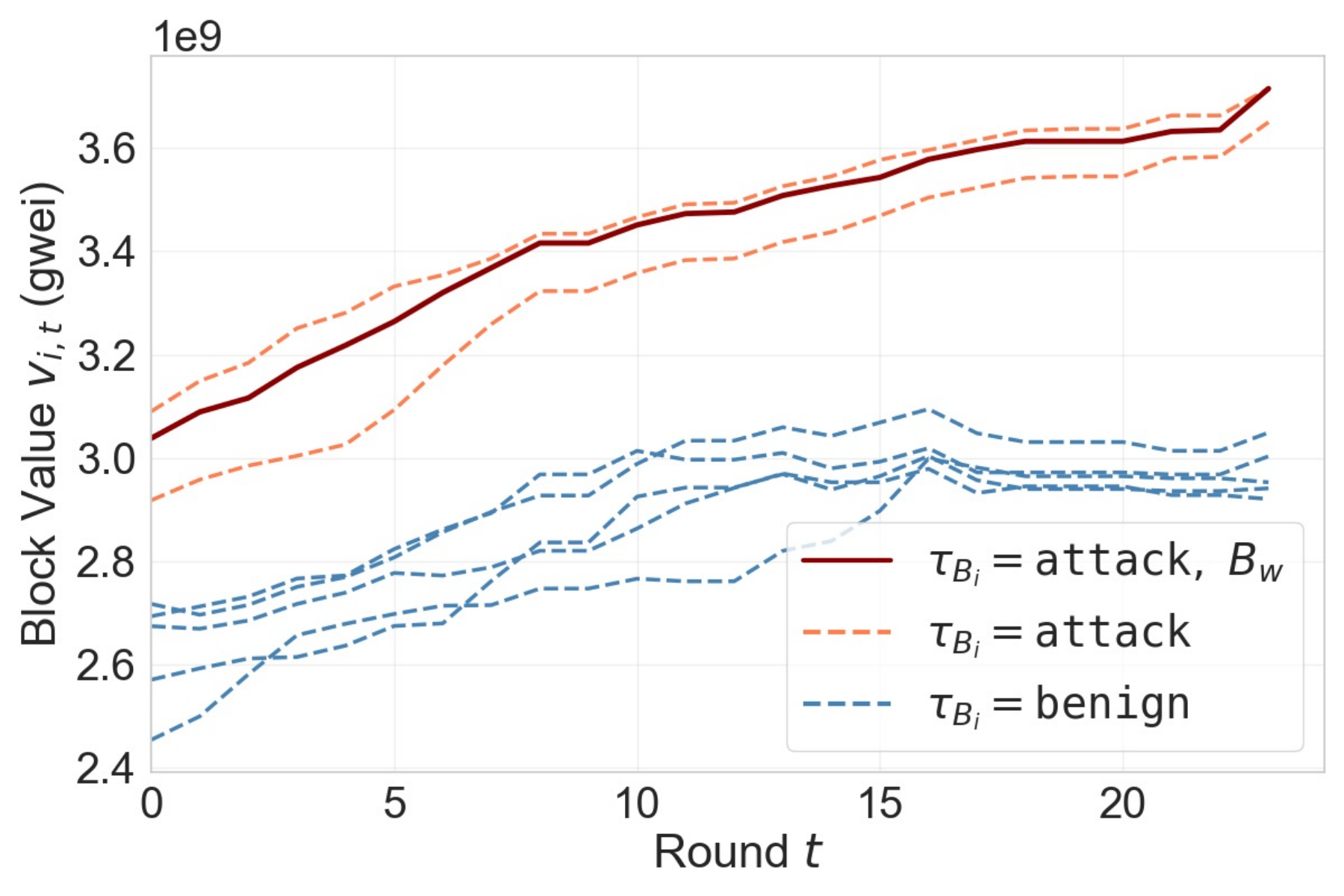}
        \caption{Block valuations \( v_{i,t} \) over rounds.}
        \label{fig:bv_dynamic}
    \end{minipage}

\end{figure}

As shown in \autoref{fig:bv_dynamic}, builders' valuations \( v_{i,t} \) grow over rounds as their mempools \( \mathcal{M}_{B_i,t} \) evolve, enabling more aggressive bidding \( b_{i,t} = \sigma_i(v_{i,t}, H_{i,t}) \), reflected in \autoref{fig:bid_dynamic}. Over 1{,}000 simulated blocks, the highest-valuation builder wins in 97\% of cases, consistent with \autoref{thm:auction_equilibrium}. On average:
\(
b_{(1),T} \approx 0.954 \, v_{(1),T}, 
\quad 
b_{(1),T} \approx 1.0004 \, v_{(2),T},
\)
where \( b_{(1),T} \) is the winning bid, \( v_{(1),T} \) the top valuation, and \( v_{(2),T} \) the second-highest. This confirms near second-price auction outcomes predicted by theory.  

These results answer \textbf{RQ1.3}. By \autoref{thm:auction_equilibrium}, when bid propagation latency is sufficiently small, the ePBS auction converges to the second-highest bid. In our simulations, latency prevents full convergence, but the winning bid remains very close to the second-highest valuation, leading to proposers capturing nearly 95\% of block value while builder profits remain highly unequal (Gini = 0.8358). By contrast, randomized slot assignment yields a much more even distribution of proposer revenue (Gini = 0.1749).

\subsubsection{Restaking and Vertical Integration}
\label{sec:restaking}

To study the \enquote{rich get richer} effect beyond single-block dynamics, we extend the analysis to 10,000 slots with all participants reinvesting profits. Builders, proposers, and validators restake, creating vertical integration. Initial stakes are distributed so that most participants hold one validator, with a few starting at higher levels.

\begin{figure}[tb]
    \centering
    \begin{subfigure}[t]{0.32\textwidth}
        \centering
        \includegraphics[width=\linewidth]{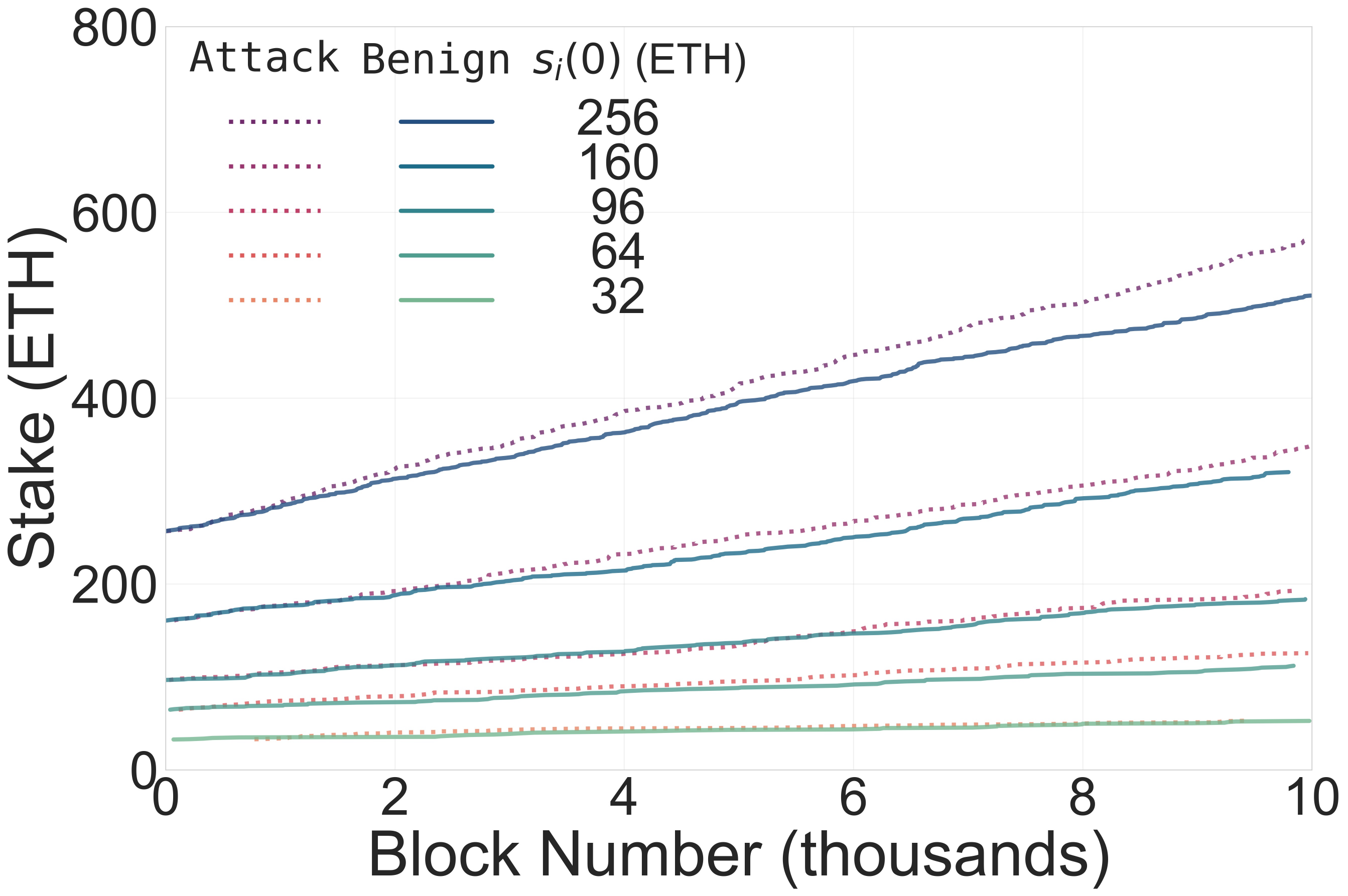}
        \caption{PoS Validator.}
        \label{fig:validator_stake}
    \end{subfigure}
    \begin{subfigure}[t]{0.32\textwidth}
        \centering
        \includegraphics[width=\linewidth]{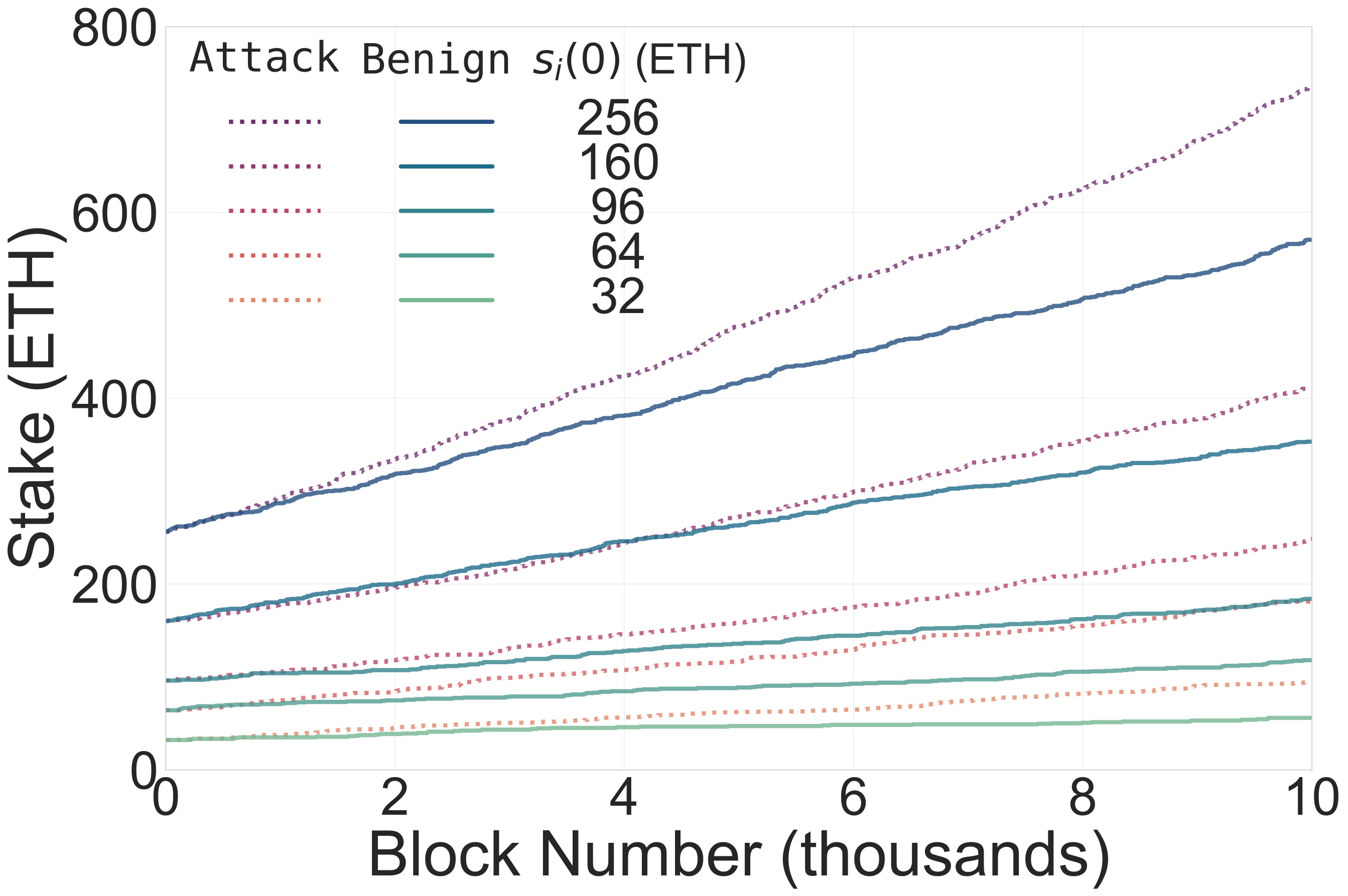}
        \caption{ePBS Builder.}
        \label{fig:builder_stake}
    \end{subfigure}
    \begin{subfigure}[t]{0.32\textwidth}
        \centering
        \includegraphics[width=\linewidth]{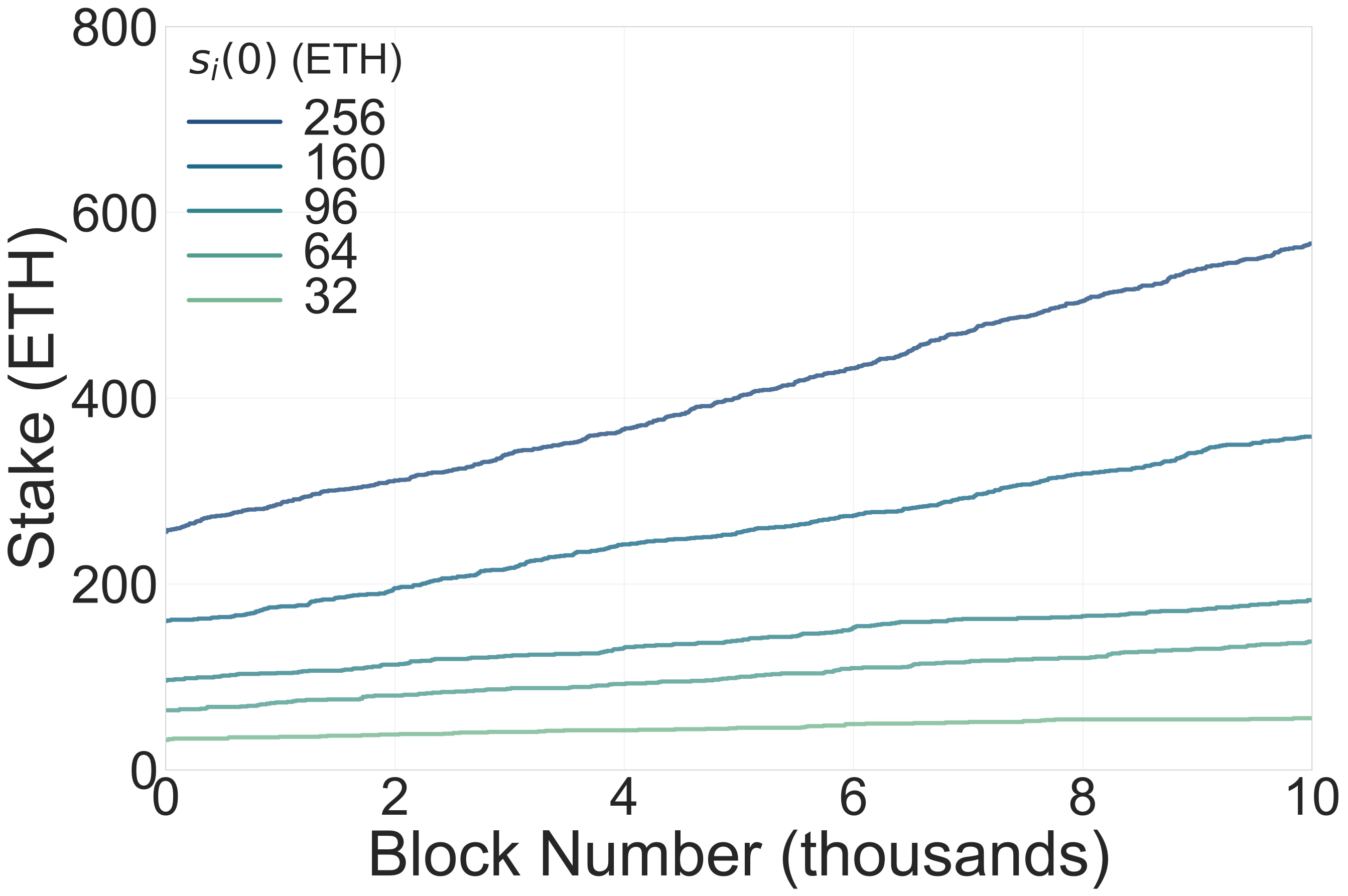}
        \caption{ePBS Proposer.}
        \label{fig:proposer_stake}
    \end{subfigure}
    \caption{Stake evolution over 10,000 blocks for participants with different initial stakes.}

    \label{fig:stake}
\end{figure}

\autoref{fig:stake} shows the resulting stake dynamics. Under \ac{pos}, validators grow steadily and higher-stake participants compound faster. Under \ac{epbs}, builders grow much more sharply thanks to dual revenues from auctions and block proposals, which widen the gap between $\mathtt{attack}$ and $\mathtt{benign}$ builders. Proposers grow at a steadier pace. In all cases, a higher initial stake leads to faster growth.

Quantitatively, participants starting with 256 ETH show distinct acceleration patterns. Attack validators reduce the blocks needed to earn 100 ETH from 2811 to 1761, while benign validators slow from 3366 to 3887. \ac{epbs} builders accelerate most: attack builders shorten from 2517 to 1664 blocks, benign builders from 3141 to 2377. Proposers instead slow slightly, from 3112 to 3304, as their relative stake declines against faster-growing builders.

These results show consistency with \autoref{sec:reinvest}, that restaking amplifies centralization more in \ac{epbs} than in \ac{pos}. Under \ac{epbs}, builders' integrated revenue streams drive faster accumulation than in \ac{pos}, deepening long-term profit and content concentration.

\subsection{RQ2 Answer: Effect of \ac{epbs} on \ac{mev} Transaction Ordering}
\label{sec:inversion}

We analyze how \ac{epbs} affects transaction ordering, since \ac{mev} activities undermine fairness and efficiency. To measure reordering, we compute the inversion count, i.e., the number of swaps needed to sort a block's transactions by creation time.

\begin{figure}[tb]
    \centering
    \begin{subfigure}[t]{0.48\textwidth}
        \centering
        \includegraphics[height=0.25\textheight]{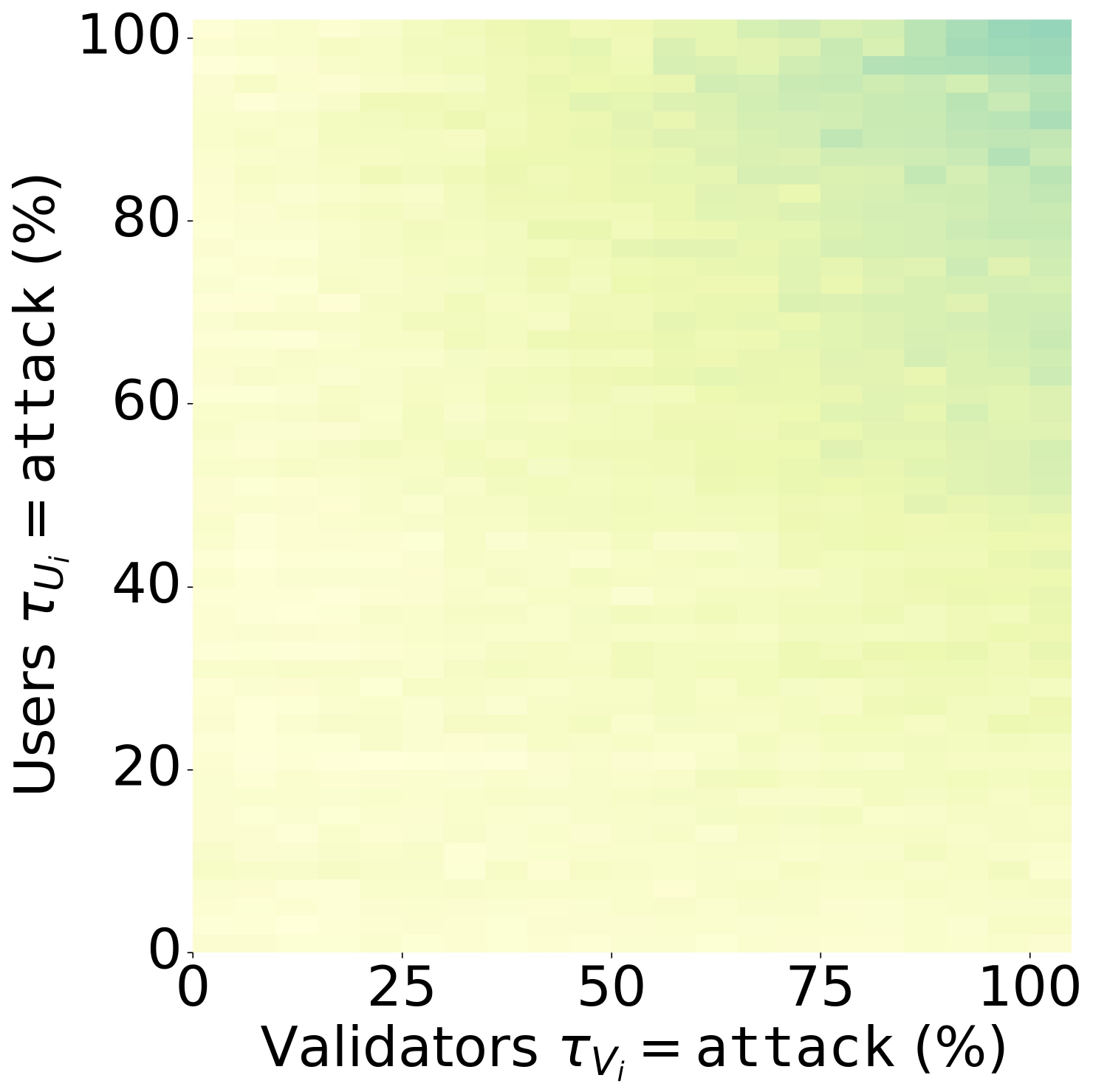}
        \caption{PoS.}
        \label{fig:inversion_pos}
    \end{subfigure}\hfill
    \begin{subfigure}[t]{0.48\textwidth}
        \centering
        \includegraphics[height=0.25\textheight]{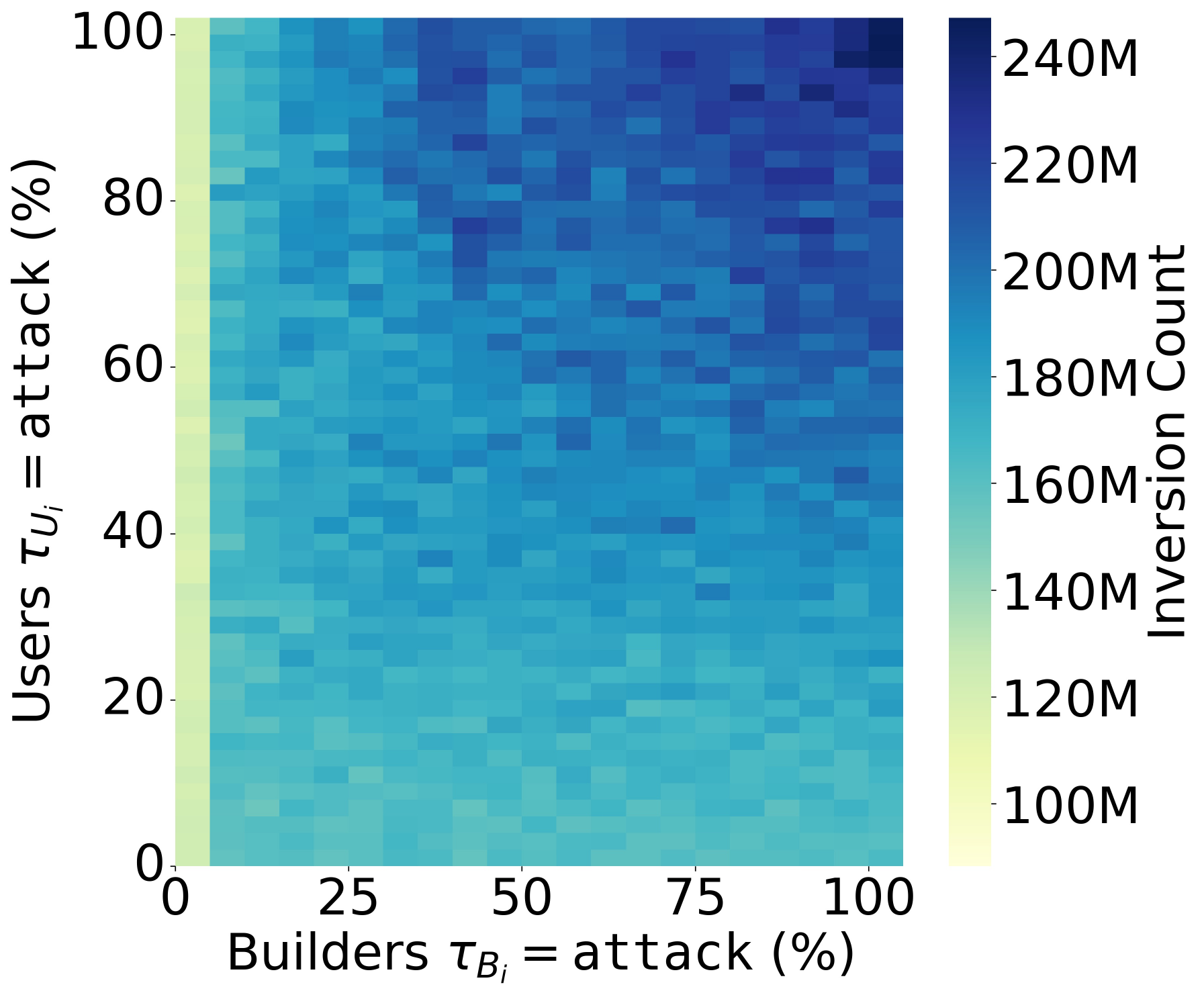}
        \caption{ePBS.}
        \label{fig:inversion_pbs}
    \end{subfigure}
    \caption{Inversion count heatmaps under varying numbers of \ac{mev}-seeking participants.}
    \label{fig:inversion}
\end{figure}

\autoref{fig:inversion} shows that more attacking participants lead to higher inversion counts in both systems, but the effect is much stronger under \ac{epbs}. In particular, \autoref{fig:inversion_pbs} indicates that the auction incentivizes aggressive reordering strategies for \ac{mev} capture, while \autoref{fig:inversion_pos} remains closer to chronological ordering.  

These findings answer \textbf{RQ2}: \ac{epbs} amplifies transaction reordering because \ac{mev}-seeking builders raise block valuations and bids by altering order. This aligns with \autoref{sec:pbs_pos_mev_expectation} and highlights added risks for fairness and decentralization.

\section{Conclusion}
\label{sec:conclusion}

Our work shows that \ac{epbs} amplifies centralization rather than mitigating it. Builder profits become highly concentrated (Gini coefficient rising from 0.1749 under \ac{pos} to 0.8358 under \ac{epbs}), and 95.4\% of block value accrues to proposers despite their limited role in block assembly. This concentration arises because \ac{mev}-seeking builders consistently outbid the benign ones, and the auction mechanism incentivizes aggressive \ac{mev} extraction.

From the builder's perspective, both profit and content building centralize sharply, with a few efficient builders dominating block construction. Proposers, by contrast, experience more even profit distribution due to randomized slot assignment, leading to partial decentralization at their level. Together, \ac{epbs} redistributes rewards but fails to decentralize block content. On transaction ordering, \ac{epbs} increases the frequency of reordering compared to \ac{pos}, as builders exploit \ac{mev} opportunities to raise bids. This dynamic strengthens the position of \ac{mev}-seeking builders, undermining fairness and further concentrating power.  

Overall, while \ac{epbs} was designed to improve decentralization, our results show that it intensifies builder dominance and transaction manipulation. Future research should investigate alternative mechanisms and designs that limit builder concentration without restoring proposer centralization.

\bibliography{references_revised1}

\end{document}

%% file: titlepage_revised1.tex
\author[a]{Yitian Wang}
\ead{tammy.wang.20@ucl.ac.uk}

\author[b]{Yebo Feng\corref{cor1}}
\ead{yebo.feng@ntu.edu.sg}

\author[c]{Yingjiu Li}
\ead{yingjiul@uoregon.edu}

\author[a]{Jiahua Xu}
\ead{jiahua.xu@ucl.ac.uk}

\affiliation[a]{
  organization = {Department of Computer Science, University College London},
  country      = {United Kingdom}
}

\affiliation[b]{
  organization = {College of Computing and Data Science, Nanyang Technological University},
  country      = {Singapore}
}

\affiliation[c]{
  organization = {Computer Science Department, University of Oregon},
  country      = {United States}
}

\cortext[cor1]{Corresponding author}

%% file: references_revised1.bib
@inproceedings{Weintraub2022,
    title = {{A Flash(bot) in the Pan: Measuring Maximal Extractable Value in Private Pools}},
    year = {2022},
    booktitle = {IMC '22: Proceedings of the 22nd ACM Internet Measurement Conference},
    author = {Weintraub, Ben and Torres, Christof Ferreira and Nita-Rotaru, Cristina and State, Radu},
    month = {10},
    pages = {458--471},
    publisher = {},
    url = {https://dl.acm.org/doi/10.1145/3517745.3561448},
    isbn = {9781450392594}
}

@article{Carbonneau2016,
    title = {{A Multi-attribute bidding strategy for a single-attribute auction marketplace}},
    year = {2016},
    journal = {Expert Systems with Applications},
    author = {Carbonneau, Réal and Vahidov, Rustam},
    month = {1},
    pages = {42--50},
    volume = {43},
    publisher = {Elsevier Ltd},
    url = {https://www.sciencedirect.com/science/article/pii/S0957417415005965},
    doi = {10.1016/j.eswa.2015.08.039},
    issn = {09574174}
}

@techreport{Zust2021,
    title = {{Analyzing and Preventing Sandwich Attacks in Ethereum}},
    year = {2021},
    author = {Z{\"{u}}st, Patrick and Nadahalli, Tejaswi and Wang Roger Wattenhofer, Ye},
    url = {https://pub.tik.ee.ethz.ch/students/2021-FS/BA-2021-07.pdf}
}

@article{Jayabalasamy2024,
    title = {{Application of Graph Theory for Blockchain Technologies}},
    year = {2024},
    journal = {Mathematics},
    author = {Jayabalasamy, Guruprakash and Pujol, Cyril and Latha Bhaskaran, Krithika},
    number = {8},
    month = {4},
    volume = {12},
    publisher = {Multidisciplinary Digital Publishing Institute (MDPI)},
    url = {https://www.mdpi.com/2227-7390/12/8/1133},
    doi = {10.3390/math12081133},
    issn = {22277390}
}

@article{Wahrstatter_Censorship,
    title = {{Blockchain Censorship}},
    year = {2024},
    journal = {WWW '24: Proceedings of the ACM Web Conference 2024},
    author = {Wahrst{\"{a}}tter, Anton and Ernstberger, Jens and Yaish, Aviv and Zhou, Liyi and Qin, Kaihua and Tsuchiya, Taro and Steinhorst, Sebastian and Svetinovic, Davor and Christin, Nicolas and Barczentewicz, Mikolaj and Gervais, Arthur},
    month = {5},
    url = {https://dl.acm.org/doi/10.1145/3589334.3645431},
    arxivId = {2305.18545}
}

@article{Babu2022,
    title = {{Blockchain MEV minimisation solution with price guarantee}},
    year = {2022},
    journal = {TechRxiv},
    author = {Pillai, Babu},
    url = {https://doi.org/10.36227/techrxiv.21345306.v1}
}

@inproceedings{Bahrani2024,
    title = {{Centralization in Block Building and Proposer-Builder Separation}},
    year = {2024},
    booktitle = {International Conference on Financial Cryptography and Data Security},
    author = {Bahrani, Maryam and Garimidi, Pranav and Roughgarden, Tim},
    url = {https://link.springer.com/chapter/10.1007/978-3-031-78676-1_19},
    arxivId = {2401.12120v1}
}

@inproceedings{Yang2024,
    title = {{Decentralization of Ethereum's Builder Market}},
    year = {2024},
    booktitle = {2025 IEEE Symposium on Security and Privacy (SP)},
    author = {Yang, Sen and Nayak, Kartik and Zhang, Fan},
    month = {5},
    url = {https://ieeexplore.ieee.org/document/11023282},
    arxivId = {2405.01329}
}

@inproceedings{Huang2021,
    title = {{Do the Rich Get Richer? Fairness Analysis for Blockchain Incentives}},
    year = {2021},
    booktitle = {Proceedings of the ACM SIGMOD International Conference on Management of Data},
    author = {Huang, Yuming and Tang, Jing and Cong, Qianhao and Lim, Andrew and Xu, Jianliang},
    pages = {790--803},
    publisher = {Association for Computing Machinery},
    doi = {10.1145/3448016.3457285},
    issn = {07308078},
    arxivId = {2103.14713}
}

@misc{eip7732,
    title = {{EIP-7732: Enshrined Proposer-Builder Separation}},
    year = {2024},
    author = {{Ethereum Improvement Proposals}},
    url = {https://eips.ethereum.org/EIPS/eip-7732}
}

@misc{eip7732discussion,
    title = {{EIP-7732: Enshrined Proposer-Builder Separation (ePBS)}},
    year = {2024},
    author = {{Ethereum Magicians}}
}

@article{Ranganatha2025,
    title = {{Enhancing fraud detection efficiency in mobile transactions through the integration of bidirectional 3d Quasi-Recurrent Neural network and blockchain technologies}},
    year = {2025},
    journal = {Expert Systems with Applications},
    author = {Ranganatha, H. R. and Syed Mustafa, A.},
    month = {1},
    volume = {260},
    publisher = {Elsevier Ltd},
    url = {https://www.sciencedirect.com/science/article/pii/S0957417424020463},
    doi = {10.1016/j.eswa.2024.125179},
    issn = {09574174}
}

@unpublished{Grandjean2023,
    title = {{Ethereum Proof-of-Stake Consensus Layer: Participation and Decentralization}},
    year = {2023},
    author = {Grandjean, Dominic and Heimbach, Lioba and Wattenhofer, Roger},
    month = {6},
    url = {http://arxiv.org/abs/2306.10777},
    arxivId = {2306.10777}
}

@article{Heimbach2023,
    title = {{Ethereum's Proposer-Builder Separation: Promises and Realities}},
    year = {2023},
    journal = {IMC '23: Proceedings of the 2023 ACM on Internet Measurement Conference},
    author = {Heimbach, Lioba and Kiffer, Lucianna and Torres, Christof Ferreira and Wattenhofer, Roger},
    month = {5},
    url = {https://dl.acm.org/doi/10.1145/3618257.3624824},
    arxivId = {2305.19037}
}

@unpublished{Daian2019,
    title = {{Flash Boys 2.0: Frontrunning, Transaction Reordering, and Consensus Instability in Decentralized Exchanges}},
    year = {2019},
    author = {Daian, Philip and Goldfeder, Steven and Kell, Tyler and Li, Yunqi and Zhao, Xueyuan and Bentov, Iddo and Breidenbach, Lorenz and Juels, Ari},
    month = {4},
    url = {http://arxiv.org/abs/1904.05234},
    arxivId = {1904.05234}
}

@misc{flashbot,
    title = {{Flashbot Docs}},
    year = {2024},
    booktitle = {https://docs.flashbots.net/flashbots-auction/overview},
    author = {{Flashbot}},
    url = {https://docs.flashbots.net/flashbots-auction/overview}
}

@misc{Buterin2022,
    title = {{How much can we constrain builders without bringing back heavy burdens to proposers?}},
    year = {2022},
    booktitle = {https://ethresear.ch/t/how-much-can-we-constrain-builders-without-bringing-back-heavy-burdens-to-proposers/13808},
    author = {Buterin, Vitalik},
    url = {https://ethresear.ch/t/how-much-can-we-constrain-builders-without-bringing-back-heavy-burdens-to-proposers/13808}
}

@inproceedings{Chitra2022,
    title = {{Improving Proof of Stake Economic Security via MEV Redistribution}},
    year = {2022},
    booktitle = {Proceedings of the 2022 ACM CCS Workshop on Decentralized Finance and Security},
    author = {Chitra, Tarun and Kulkarni, Kshitij},
    month = {11},
    pages = {1--7},
    publisher = {Association for Computing Machinery, Inc},
    url = {https://dl.acm.org/doi/10.1145/3560832.3564259},
    isbn = {9781450385404}
}

@misc{ethmev,
    title = {{MAXIMAL EXTRACTABLE VALUE (MEV)}},
    year = {2023},
    booktitle = {https://ethereum.org/en/developers/docs/mev/},
    author = {{Corwin Smith}},
    url = {https://ethereum.org/en/developers/docs/mev/}
}

@misc{mev_distribution,
    title = {{Modelling Realised Extractable Value in Proof of Stake Ethereum}},
    year = {2022},
    booktitle = {https://collective.flashbots.net/t/modelling-realised-extractable-value-in-proof-of-stake-ethereum/290},
    author = {{Elainehu}},
    url = {https://collective.flashbots.net/t/modelling-realised-extractable-value-in-proof-of-stake-ethereum/290}
}

@unpublished{Jensen2023,
    title = {{Multi-block MEV}},
    year = {2023},
    author = {Jensen, Johannes Rude and von Wachter, Victor and Ross, Omri},
    month = {3},
    url = {http://arxiv.org/abs/2303.04430},
    arxivId = {2303.04430}
}

@inproceedings{zhou_arbitrage,
    title = {{On the just-in-time discovery of profit-generating transactions in DeFi Protocols}},
    year = {2021},
    booktitle = {2021 IEEE Symposium on Security and Privacy},
    author = {Zhou, Liyi and Qin, Kaihua and Cully, Antoine and Livshits, Benjamin and Gervais, Arthur},
    month = {5},
    pages = {919--936},
    volume = {2021-May},
    publisher = {Institute of Electrical and Electronics Engineers Inc.},
    url = {https://www.computer.org/csdl/proceedings-article/sp/2021/893400b955/1t0x9Rx2te8},
    isbn = {9781728189345},
    issn = {10816011},
    arxivId = {2103.02228}
}

@inproceedings{Wang2024,
    title = {{Private Order Flows and Builder Bidding Dynamics: The Road to Monopoly in Ethereum's Block Building Market}},
    year = {2025},
    booktitle = {WWW '25: Proceedings of the ACM on Web Conference 2025},
    author = {Wang, Shuzheng and Huang, Yue and Zhang, Wenqin and Huang, Yuming and Wang, Xuechao and Tang, Jing},
    month = {10},
    url = {https://doi.org/10.1145/3696410.3714754},
    arxivId = {2410.12352}
}

@misc{pbsdef,
    title = {{Proposer-builder separation}},
    year = {2024},
    booktitle = {https://ethereum.org/nl/roadmap/pbs/},
    author = {{Pagina}},
    url = {https://ethereum.org/nl/roadmap/pbs/}
}

@unpublished{Capponi2024,
    title = {{Proposer-Builder Separation, Payment for Order Flows, and Centralization in Blockchain}},
    year = {2024},
    author = {Capponi, Agostino and Jia, Ruizhe and Olafsson, Sveinn},
    url = {https://ssrn.com/abstract=4723674}
}

@article{Qin2021,
    title = {{Quantifying Blockchain Extractable Value: How dark is the forest?}},
    year = {2022},
    journal = {2022 IEEE Symposium on Security and Privacy},
    author = {Qin, Kaihua and Zhou, Liyi and Gervais, Arthur},
    month = {1},
    url = {https://ieeexplore.ieee.org/document/9833734/},
    arxivId = {2101.05511}
}

@misc{2subsec,
    title = {{Relay API Documentation}},
    year = {2024},
    booktitle = {https://docs.flashbots.net/flashbots-mev-boost/relay},
    author = {{Flashbot}},
    url = {https://docs.flashbots.net/flashbots-mev-boost/relay}
}

@inproceedings{yang2022,
    title = {{SoK: MEV Countermeasures: Theory and Practice}},
    year = {2022},
    booktitle = {DeFi '24: Proceedings of the Workshop on Decentralized Finance and Security},
    author = {Yang, Sen and Zhang, Fan and Huang, Ken and Chen, Xi and Yang, Youwei and Zhu, Feng},
    month = {12},
    url = {https://dl.acm.org/doi/10.1145/3689931.3694911},
    arxivId = {2212.05111}
}

@article{Eskandari2019,
    title = {{SoK: Transparent Dishonesty: front-running attacks on Blockchain}},
    year = {2019},
    journal = {Financial Cryptography and Data Security},
    author = {Eskandari, Shayan and Moosavi, Seyedehmahsa and Clark, Jeremy},
    pages = {170--189},
    url = {http://arxiv.org/abs/1902.05164},
    arxivId = {1902.05164}
}

@misc{ethpbs,
    title = {{State of research: increasing censorship resistance of transactions under proposer/builder separation (PBS)}},
    year = {2021},
    author = {{Vitalik Buterin}},
    url = {https://notes.ethereum.org/@vbuterin/pbs_censorship_resistance}
}

@inproceedings{Wu2023,
    title = {{Strategic Bidding Wars in On-chain Auctions}},
    year = {2023},
    author = {Wu, Fei and Thiery, Thomas and Leonardos, Stefanos and Ventre, Carmine},
    month = {12},
    publisher = {Proceedings of the 6th edition of the IEEE International Conference on Blockchain and Cryptocurrency (ICBC 2024)},
    url = {https://ieeexplore.ieee.org/document/10634354/},
    arxivId = {2312.14510}
}

@article{Gupta2023,
    title = {{The centralizing effects of private order flow on proposer-builder separation}},
    year = {2023},
    journal = {5th Conference on Advances in Financial Technologies (AFT 2023)},
    author = {Gupta, Tivas and Pai, Mallesh M and Resnick, Max},
    month = {5},
    url = {http://arxiv.org/abs/2305.19150},
    arxivId = {2305.19150}
}

@misc{merge,
    title = {{The Merge}},
    year = {2024},
    booktitle = {https://ethereum.org/en/roadmap/merge/},
    author = {{Ethereum}},
    url = {https://ethereum.org/en/roadmap/merge/}
}

@misc{Leonardo2022,
    title = {{The risks of vertical integration in MEV-Boost}},
    year = {2022},
    booktitle = {Flashbot},
    author = {{Leonardo Arias Fonseca}},
    url = {https://collective.flashbots.net/t/the-risks-of-vertical-integration-in-mev-boost/235}
}

@article{Wahrstatter2023,
    title = {{Time to Bribe: Measuring Block Construction Market}},
    year = {2023},
    journal = {Preprint},
    author = {Wahrst{\"{a}}tter, Anton and Zhou, Liyi and Qin, Kaihua and Svetinovic, Davor and Gervais, Arthur},
    month = {5},
    url = {http://arxiv.org/abs/2305.16468},
    arxivId = {2305.16468}
}

@unpublished{Makarov2019,
    title = {{Trading and Arbitrage in Cryptocurrency Markets}},
    year = {2019},
    author = {Makarov, Igor and Schoar, Antoinette},
    url = {https://ssrn.com/abstract=3171204}
}

@misc{Neuder2023_epbs,
    title = {{Why enshrine Proposer-Builder Separation? A viable path to ePBS}},
    year = {2023},
    author = {Neuder, Michael},
    url = {https://ethresear.ch/t/why-enshrine-proposer-builder-separation-a-viable-path-to-epbs/15710/1}
}

@article{Ntuala2026ADetection,
    title = {{A dual-layer GNN with economic penalty mechanisms for blockchain fraud detection}},
    year = {2026},
    journal = {Expert Systems with Applications},
    author = {Ntuala, Grace Mupoyi and Xia, Qi and Xia, Hu and Badjie, Ansu and Mukala, Patrick and Da Silva Tavares, Edson Eliezer and Gao, Jianbin and Ukwuoma, Chiagoziem C.},
    month = {3},
    pages = {130121},
    volume = {299},
    publisher = {Elsevier BV},
    url = {https://www.sciencedirect.com/science/article/pii/S0957417425037364},
    doi = {10.1016/j.eswa.2025.130121},
    issn = {09574174}
}

@article{Mcafee1987AuctionsBidding,
    title = {{Auctions and Bidding}},
    year = {1987},
    journal = {Journal of Economic Literature},
    author = {Mcafee, R Preston and Mcmillan, John},
    number = {2},
    pages = {699--738},
    volume = {25},
    url = {https://cramton.umd.edu/econ415/mcafee-mcmillan-auctions-and-bidding-jel-1987.pdf}
}

@misc{Neuder2023BidHarmful,
    title = {{Bid cancellations considered harmful}},
    year = {2023},
    author = {Neuder, Michael},
    url = {https://ethresear.ch/t/bid-cancellations-considered-harmful/15500?utm_source=chatgpt.com}
}

@inproceedings{Motepalli2025DecentralizationAdvancement,
    title = {{Decentralization in PoS Blockchain Consensus: Quantification and Advancement}},
    year = {2025},
    booktitle = {IEEE Transactions on Network and Service Management ( Volume: 22, Issue: 4, August 2025)},
    author = {Motepalli, Shashank and Jacobsen, Hans-Arno},
    month = {4},
    url = {https://ieeexplore.ieee.org/document/10965870/},
    arxivId = {2504.14351}
}

@misc{ThomasThiery2023EmpiricalBBPs,
    title = {{Empirical analysis of Builders’ Behavioral Profiles (BBPs)}},
    year = {2023},
    author = {{Thomas Thiery}},
    url = {https://ethresear.ch/t/empirical-analysis-of-builders-behavioral-profiles-bbps/16327/1}
}

@article{Jeong2026LiqBoost:Exchanges,
    title = {{LiqBoost: Enhancing liquidity provision for blockchain-based decentralized exchanges}},
    year = {2026},
    journal = {Expert Systems with Applications},
    author = {Jeong, Woojin and Park, Seongwan and Lee, Jaewook and Lee, Yunyoung},
    month = {2},
    volume = {297},
    publisher = {Elsevier Ltd},
    url = {https://www.sciencedirect.com/science/article/pii/S0957417425028465},
    doi = {10.1016/j.eswa.2025.129230},
    issn = {09574174}
}

@misc{Beaconscan2025SkippedBlocks,
    title = {{Skipped Blocks}},
    year = {2025},
    author = {{Beaconscan}},
    url = {https://beaconscan.com/slots-skipped}
}

@inproceedings{Werner2022SoK:DeFi,
    title = {{SoK: Decentralized Finance (DeFi)}},
    year = {2022},
    booktitle = {AFT '22: Proceedings of the 4th ACM Conference on Advances in Financial Technologies},
    author = {Werner, Sam and Perez, Daniel and Gudgeon, Lewis and Klages-Mundt, Ariah and Harz, Dominik and Knottenbelt, William},
    month = {9},
    pages = {30--46},
    publisher = {Association for Computing Machinery (ACM)},
    url = {https://doi.org/10.1145/3558535.3559780}
}

@inproceedings{Li2021TopoShot:Transactions,
    title = {{TopoShot: Uncovering Ethereum's network topology leveraging replacement transactions}},
    year = {2021},
    booktitle = {IMC '21: Proceedings of the 21st ACM Internet Measurement Conference},
    author = {Li, Kai and Tang, Yuzhe and Chen, Jiaqi and Wang, Yibo and Liu, Xianghong},
    month = {11},
    pages = {302--319},
    publisher = {},
    url = {https://dl.acm.org/doi/10.1145/3487552.3487814},
    isbn = {9781450391290},
    arxivId = {2109.14794}
}

@unpublished{Beccuti2025TowardsMarket,
    title = {{Towards a Formal Framework of the Ethereum Staking Market}},
    year = {2025},
    booktitle = {arXiv},
    author = {Beccuti, Juan and Chantramonklasri, Thunj and Hafner, Matthias and Oderbolz, Nicolas},
    month = {7},
    url = {http://arxiv.org/abs/2503.14385},
    arxivId = {2503.14385}
}

@article{Oz2024WhoWhy,
    title = {{Who Wins Ethereum Block Building Auctions and Why?}},
    year = {2024},
    author = {{\"{O}}z, Burak and Sui, Danning and Thiery, Thomas and Matthes, Florian},
    month = {7},
    url = {http://arxiv.org/abs/2407.13931},
    arxivId = {2407.13931}
}
